\documentclass[11pt]{article}

\usepackage{amsmath,amssymb,amsthm,amsfonts,latexsym,mathtools,mathdots,bbm,graphicx,float}
\usepackage{bbold}
\usepackage[letterpaper,margin=1in]{geometry}
\usepackage[backref,colorlinks,citecolor=blue,bookmarks=true]{hyperref}
\usepackage[nameinlink]{cleveref}
\usepackage{mdframed}
\usepackage[dvipsnames]{xcolor}
\usepackage{thmtools}
\usepackage{thm-restate}

\title{Amortized Locally Decodable Codes for Insertions and Deletions}
\author{Jeremiah Blocki \and Justin Zhang}
\date{}

\begin{document}
\newcommand{\todo}[1]{\textcolor{red}{\underline{Todo}: #1}}
\newcommand{\jadd}[1]{\textcolor{PineGreen} {#1}}
\newcommand{\justin}[1]{\textcolor{blue} {\underline{Justin says:} #1}}
\newcommand{\jrem}[1]{\textcolor{purple} {\underline{Justin removed:} #1}}

\newcommand{\new}[1]{\textcolor{MidnightBlue}{#1}}
\newcommand{\jeremiah}[1]{\textcolor{red}{#1}}

\newcommand{\msf}[1]{\ensuremath{\mathsf{#1}}}
\newcommand{\mc}{\mathcal}

\newcommand{\calA}{\mathcal{A}}
\newcommand{\calB}{\mathcal{B}}
\newcommand{\calC}{\mathcal{C}}
\newcommand{\calD}{\mathcal{D}}
\newcommand{\calE}{\mathcal{E}}
\newcommand{\calF}{\mathcal{F}}
\newcommand{\calG}{\mathcal{G}}
\newcommand{\calH}{\mathcal{H}}
\newcommand{\calI}{\mathcal{I}}
\newcommand{\calJ}{\mathcal{J}}
\newcommand{\calK}{\mathcal{K}}
\newcommand{\calL}{\mathcal{L}}
\newcommand{\calM}{\mathcal{M}}
\newcommand{\calN}{\mathcal{N}}
\newcommand{\calO}{\mathcal{O}}
\newcommand{\calP}{\mathcal{P}}
\newcommand{\calQ}{\mathcal{Q}}
\newcommand{\calR}{\mathcal{R}}
\newcommand{\calS}{\mathcal{S}}
\newcommand{\calT}{\mathcal{T}}
\newcommand{\calU}{\mathcal{U}}
\newcommand{\calV}{\mathcal{V}}
\newcommand{\calW}{\mathcal{W}}
\newcommand{\calX}{\mathcal{X}}
\newcommand{\calY}{\mathcal{Y}}
\newcommand{\calZ}{\mathcal{Z}}

\newcommand{\R}{\mathbbm R}
\newcommand{\C}{\mathbbm C}
\newcommand{\N}{\mathbbm N}
\newcommand{\Z}{\mathbbm Z}
\newcommand{\F}{\mathbbm F}
\newcommand{\Q}{\mathbbm Q}

\newcommand{\E}{\mathbbm E}

\newcommand{\bone}{\boldsymbol{1}}
\newcommand{\bbeta}{\boldsymbol{\beta}}
\newcommand{\bdelta}{\boldsymbol{\delta}}
\newcommand{\bepsilon}{\boldsymbol{\epsilon}}
\newcommand{\blambda}{\boldsymbol{\lambda}}
\newcommand{\bomega}{\boldsymbol{\omega}}
\newcommand{\bpi}{\boldsymbol{\pi}}
\newcommand{\bphi}{\boldsymbol{\phi}}
\newcommand{\bvphi}{\boldsymbol{\varphi}}
\newcommand{\bpsi}{\boldsymbol{\psi}}
\newcommand{\bsigma}{\boldsymbol{\sigma}}
\newcommand{\btheta}{\boldsymbol{\theta}}
\newcommand{\btau}{\boldsymbol{\tau}}
\newcommand{\ba}{\boldsymbol{a}}
\newcommand{\bb}{\boldsymbol{b}}
\newcommand{\bc}{\boldsymbol{c}}
\newcommand{\bd}{\boldsymbol{d}}
\newcommand{\be}{\boldsymbol{e}}
\newcommand{\boldf}{\boldsymbol{f}}
\newcommand{\bg}{\boldsymbol{g}}
\newcommand{\bh}{\boldsymbol{h}}
\newcommand{\bi}{\boldsymbol{i}}
\newcommand{\bj}{\boldsymbol{j}}
\newcommand{\bk}{\boldsymbol{k}}
\newcommand{\bell}{\boldsymbol{\ell}}
\newcommand{\bm}{\boldsymbol{m}}
\newcommand{\bn}{\boldsymbol{n}}
\newcommand{\bo}{\boldsymbol{o}}
\newcommand{\bp}{\boldsymbol{p}}
\newcommand{\bq}{\boldsymbol{q}}
\newcommand{\br}{\boldsymbol{r}}
\newcommand{\bs}{\boldsymbol{s}}
\newcommand{\bt}{\boldsymbol{t}}
\newcommand{\bu}{\boldsymbol{u}}
\newcommand{\bv}{\boldsymbol{v}}
\newcommand{\bw}{\boldsymbol{w}}
\newcommand{\bx}{{\boldsymbol{x}}}
\newcommand{\by}{\boldsymbol{y}}
\newcommand{\bz}{\boldsymbol{z}}
\newcommand{\bA}{\boldsymbol{A}}
\newcommand{\bB}{\boldsymbol{B}}
\newcommand{\bC}{\boldsymbol{C}}
\newcommand{\bD}{\boldsymbol{D}}
\newcommand{\bE}{\boldsymbol{E}}
\newcommand{\bF}{\boldsymbol{F}}
\newcommand{\bG}{\boldsymbol{G}}
\newcommand{\bH}{\boldsymbol{H}}
\newcommand{\bI}{\boldsymbol{I}}
\newcommand{\bJ}{\boldsymbol{J}}
\newcommand{\bK}{\boldsymbol{K}}
\newcommand{\bL}{\boldsymbol{L}}
\newcommand{\bM}{\boldsymbol{M}}
\newcommand{\bN}{\boldsymbol{N}}
\newcommand{\bP}{\boldsymbol{P}}
\newcommand{\bQ}{\boldsymbol{Q}}
\newcommand{\bR}{\boldsymbol{R}}
\newcommand{\bS}{\boldsymbol{S}}
\newcommand{\bT}{\boldsymbol{T}}
\newcommand{\bU}{\boldsymbol{U}}
\newcommand{\bV}{\boldsymbol{V}}
\newcommand{\bW}{\boldsymbol{W}}
\newcommand{\bX}{\boldsymbol{X}}
\newcommand{\bY}{\boldsymbol{Y}}
\newcommand{\bZ}{\boldsymbol{Z}}

\newcommand{\unif}{\overset{{\scriptscriptstyle\$}}{\leftarrow}}

\newcommand{\poly}{\text{poly}}
\newcommand{\polylog}{\text{polylog}}
\newcommand{\negl}{\text{negl}}

\newcommand{\Gen}{\ensuremath{\msf{Gen}}}
\newcommand{\Enc}{\ensuremath{\msf{Enc}}}
\newcommand{\Dec}{\ensuremath{\msf{Dec}}}
\newcommand{\sk}{\ensuremath{\msf{sk}}}
\newcommand{\pk}{\ensuremath{\msf{pk}}}

\makeatletter
\newcommand{\subalign}[1]{%
  \vcenter{%
    \Let@ \restore@math@cr \default@tag
    \baselineskip\fontdimen10 \scriptfont\tw@
    \advance\baselineskip\fontdimen12 \scriptfont\tw@
    \lineskip\thr@@\fontdimen8 \scriptfont\thr@@
    \lineskiplimit\lineskip
    \ialign{\hfil$\m@th\scriptstyle##$&$\m@th\scriptstyle{}##$\hfil\crcr
      #1\crcr
    }%
  }%
}
\makeatother

\newenvironment{weirdFrame}[1]
  {\mdfsetup{
    frametitle={\colorbox{white}{\space#1\space}},
    innertopmargin=0pt,
    skipabove=\topskip, skipbelow=\topskip,
    nobreak=true,
    frametitleaboveskip=-\ht\strutbox,
    frametitlealignment=\center
    }
  \begin{mdframed}%
  }
  {\end{mdframed}}

\newenvironment{weirdoFrame}[1]
  {\mdfsetup{
    frametitle={\colorbox{white}{\space#1\space}},
    innertopmargin=0pt,
    skipabove=\topskip, skipbelow=\topskip,
    nobreak=true,
    frametitleaboveskip=-\ht\strutbox,
    frametitlealignment=\center
    }
  \begin{mdframed}%
  }
  {\end{mdframed}}
\newcommand{\eps}{\varepsilon}
\newcommand{\inprod}[1]{\left\langle #1 \right\rangle}
\newcommand{\m}[1]{\begin{bmatrix}#1\end{bmatrix}}

\newcommand{\LDC}{\ensuremath{\msf{LDC}}}
\newcommand{\aLDC}{\ensuremath{\msf{aLDC}}}
\newcommand{\pLDC}{\ensuremath{\msf{pLDC}}}
\newcommand{\LCC}{\ensuremath{\msf{LCC}}}
\newcommand{\ham}{\msf{Ham}}
\newcommand{\Ham}{\ham}

\newcommand{\ed}{\msf{ED}}
\newcommand{\ED}{\ed}

\newcommand{\aldc}{\msf{aLDC}}
\newcommand{\paldc}{\msf{paLDC}}
\newcommand{\raldc}{\msf{raLDC}}
\newcommand{\paLDC}{\paldc}
\newcommand{\paldcgame}{\texttt{paLDC-Sec-Game}}
\newcommand{\raldcgame}{\texttt{raLDC-Game}}

\newcommand{\out}[1]{#1_{\text{out}}}
\newcommand{\inn}[1]{#1_{\text{in}}}
\newcommand{\insdel}[1]{#1_{\text{insdel}}}
\newcommand{\h}[1]{#1_{\text{h}}}
\newcommand{\p}[1]{#1_{\msf{p}}}
\newcommand{\res}[1]{#1_{\msf{r}}}
\newcommand{\calCout}{\out \calC}
\newcommand{\calCin}{\inn \calC}
\newcommand{\alphaout}{\out \alpha}
\newcommand{\alphain}{\inn \alpha}
\newcommand{\kappaout}{\out \kappa}
\newcommand{\kappain}{\inn \kappa}
\newcommand{\deltain}{\inn \delta}
\newcommand{\deltaout}{\out \delta}
\newcommand{\epsout}{\out \eps}
\newcommand{\epsin}{\inn \eps}
\newcommand{\betain}{\inn \beta}

\newcommand{\sz}[1]{#1_{\text{sz}}}
\newcommand{\calCsz}{\sz \calC}
\newcommand{\Encsz}{\sz \Enc}
\newcommand{\Decsz}{\sz \Dec}

\newcommand{\buf}{\text{buf}}
\newcommand{\NBS}{\texttt{Noisy-Binary-Search}}
\newcommand{\IntDec}{\texttt{Interval-Decode}}
\newcommand{\BlockDec}{\texttt{Block-Decode}}
\newcommand{\bufs}{\buf s}

\newcommand{\BufFind}{\texttt{Buff-Find}}

\newcommand{\Compile}{\texttt{Compile}}
\newcommand{\EncCompile}{\msf{EncCompile}}
\newcommand{\Extract}{\msf{Extract}}
\newcommand{\DecCompile}{\texttt{DecCompile}}

\newcommand{\RecoverBlock}{\texttt{RecoverBlock}}
\newcommand{\FindBlock}{\texttt{RecoverBlock}}
\newcommand{\FindBlocks}{\texttt{RecoverBlocks}}
\newcommand{\RecoverBlocks}{\texttt{RecoverBlocks}}
\newcommand{\Search}{\texttt{Search}}
\newcommand{\MarkGood}{\texttt{MarkGood}}
\newcommand{\BuffFind}{\texttt{BuffFind}}

\newcommand{\Good}{\msf{Good}}

\newcommand{\uDec}{\ensuremath{\msf{uDec}}}
\newcommand{\LDCwUGame}{\texttt{LDCwU-Game}}
\newcommand{\LDCwU}{\msf{LDCwU}}

\newcommand{\pLDCwUGame}{\texttt{pLDC-w-Uncert-Game}}
\newcommand{\pLDCwU}{\msf{pLDCwU}}

\newcommand{\findGood}{\texttt{findGood}}

\newcommand{\BI}{\calB\calI}

\newcommand{\size}[1]{\left| #1 \right|}
\newcommand{\simu}[1]{#1_{\msf{sim}}}

\newcommand{\simuu}[1]{#1_{\msf{sim,u}}}

\newcommand{\uDecGame}{\texttt{UDec-Game}}

\newcommand{\Bad}{\texttt{Bad}}

\newcommand{\reduce}{\texttt{reduce}}

\newcommand{\psk}{\msf{p},\sk}

\newcommand{\puzzgen}{\msf{PuzzGen}}
\newcommand{\PuzzGen}{\msf{PuzzGen}}
\newcommand{\puzzsolve}{\msf{PuzzSolve}}
\newcommand{\PuzzSolve}{\msf{PuzzSolve}}
\newcommand{\Sol}{\msf{PuzzSolve}}
\newcommand{\Puzz}{\msf{Puzz}}

\newcommand{\LdcStar}{\ensuremath{\LDC^*}}
\maketitle
\newtheorem{theorem}{Theorem}
\newtheorem*{remark}{Remark}
\newtheorem{lemma}[theorem]{Lemma}
\newtheorem{corollary}[theorem]{Corollary}
\newtheorem{definition}[theorem]{Definition}
\newtheorem{proposition}[theorem]{Proposition}

\newtheorem*{informal}{Theorem}

\theoremstyle{definition}
\newtheorem{construction}[theorem]{Construction}

\begin{abstract}
Locally Decodable Codes (\LDC s) are error correcting codes which permit the recovery of any single message symbol with a low number of queries to the codeword (the locality).  Traditional \LDC \ tradeoffs between the rate, locality, and error tolerance are undesirable even in relaxed settings where the encoder/decoder share randomness or where the channel is resource-bounded. Recent work by Blocki and Zhang initiated the study of Hamming {\em amortized} Locally Decodable Codes (\aLDC s), which allow the local decoder to amortize their number of queries over the recovery of a small subset of message symbols. Surprisingly, Blocki and Zhang construct asymptotically ideal (constant rate, constant amortized locality, and constant error tolerance) {\em Hamming} \aLDC s in private-key and resource-bounded settings. While this result overcame previous barriers and impossibility results for Hamming \LDC s, it is not clear whether the techniques extend to Insdel \LDC s. Constructing Insdel \LDC s which are resilient to insertion and/or deletion errors is known to be even more challenging. For example, Gupta (STOC'24) proved that no Insdel \LDC\ with constant rate and error tolerance exists even in relaxed settings.  

Our first contribution is to provide a Hamming-to-Insdel compiler which transforms any amortized Hamming \LDC\ that satisfies a particular property (consecutive interval querying) to amortized Insdel \LDC\ while asymptotically preserving the rate, error tolerance and amortized locality. Prior Hamming-to-Insdel compilers of  Ostrovsky and Paskin-Cherniavsky (ICITS'15) and Block et al. (FSTTCS'20) worked for arbitrary Hamming \LDC s, but incurred an undesirable polylogarithmic  blow-up in the locality. Our second contribution is a construction of an ideal amortized Hamming \LDC\ which satisfies our special property  (consecutive interval querying) in the relaxed settings where the sender/receiver share randomness or where the channel is resource bounded. Taken together, we obtain ideal Insdel \aLDC s in private-key and resource-bounded settings with constant amortized locality, constant rate and constant error tolerance. This result is surprising in light of Gupta's (STOC'24) impossibility result which demonstrates a strong separation between locality and amortized locality for Insdel \LDC s. 
\end{abstract}
\section{Introduction}\label{sect:introduction}
Locally Decodable Codes (\LDC s) are a variant of error correcting codes which provide single-symbol recovery with highly efficient query complexity over a (possibly corrupted) codeword.
Specifically, a \LDC \ over alphabet $\Sigma$ is defined as a tuple of algorithms the encoding algorithm $\Enc: \Sigma^k \rightarrow \Sigma^n$ and the local decoder $\Dec^{\tilde{by}} : [k] \rightarrow \Sigma$. Here, $k$ denotes the message length,  $n$ denotes the codeword length and $\tilde{\by} \in \Sigma^n$ is a (possibly corrupted) codeword --- for any message $\bx \in \Sigma^k,$ its corresponding codeword is $\by = \Enc(\bx)$ and $\tilde{\by}$ denotes a (possibly corrupted) copy of $\by$. The local decoder $\Dec^{\tilde{\by}}(i)$ is a randomized oracle algorithm that is given an index $i \in [k] := \{1,2,\dots,k\}$ as input and has oracle access to a (possibly corrupted) codeword $\tilde{\by}$. $\Dec^{\tilde{\by}}(i)$
 attempts to output the $i$'th symbol of the message $\bx$ after making {\em at most } $\ell$ queries to the symbols of $\tilde{\by}$. We require that for any index $i\in [k]$ and {\em any} $\tilde{\by}$ that is not too far from the uncorrupted codeword $\by$ (i.e $\msf{Dist}(\by,\tilde{\by}) \leq \delta |\by|$ for a chosen distance metric $\msf{Dist}$) that $\Dec^{\tilde{\by}}(i)$ outputs the correct symbol $\bx[i]$ with probability at least $1-\eps$. Broadly, codes which satisfy the above parameters are called $(\ell,\delta,\eps)$-\LDC s.
The main parameters of interest are the {\em query locality} $\ell,$ {\em error tolerance} $\delta$ and {\em rate} $k/n.$ Ideally, we would like an \LDC \ with constant locality, constant error tolerance, and constant rate simultaneously. 


Within the classical setting of worst-case Hamming errors, i.e when the distance metric is chosen to be the Hamming distance metric and error patterns are introduced in a worst-case fashion, the trade-offs between the locality, error tolerance, and rate has been extensively studied.
Even so, achievable parameters in this setting remain undesirably sub-optimal. 
Katz and Trevisan show that any \LDC \ with constant locality $\ell \ge 2$ and constant error tolerance $\delta > 0$ must have non-constant rate \cite{STOC:KatTre00}, immediately ruling out the existence of an ideal \LDC \ for worst case Hamming errors.
Moreover, the best known constructions with constant locality and error tolerance have super-polynomial rate \cite{LDCSURVEY,STOC:Efremenko09}. 
Several relaxations have been made, such as allowing the local decoder to reject corrupted codewords \cite{ITCS:GurRamRot18,STOC:BGHSV04,SODA:ChiGurShi20,STOC:KumMon24,CCC:CohYan24}, the assumption that the encoder and decoder share a cryptographic key \cite{EPRINT:OstPanSah07,C:HemOst08,HemOstStraWoo11} that is unknown to a probabilistic polynomial time channel, or the assumption that the channel is resource-bounded in other ways \cite{SCN:AmeBloBlo22,ITC:BloKulZho20}.
Yet, even under these assumptions, we do not have an \LDC \ with constant locality, constant error tolerance, and constant rate. 
{The literature on \LDC s and relaxed variants is vast. We provide an expanded discussion of the work related to \LDC s in Appendix \ref{app:relatedWork}.}

Blocki and Zhang \cite{bloZha25} introduced the notion of an {\em amortized} locally decodable code where the total number of queries to the codeword can be amortized by the total number of message bits recovered (see Definition \ref{def:aLDC}). The local decoder for an amortized \LDC \ takes as input a range $[L,R]$ (we will assume $R - L + 1 \geq \kappa$ for a suitable parameter $\kappa$) instead of a single index $i$ and attempts to output the entire substring $\bx[L,R]$ instead of the single message symbol $x_i$. The amortized locality is $\ell/(R-L+1)$ i.e.,  the total number of queries $\ell$ divided by the number of message symbols recovered. Surprisingly, they showed how to construct ideal amortized locally decodable codes in relaxed settings where the channel is resource-bounded or where the sender/receiver share a secret key. In particular, given any (sufficiently large) interval $[L,R] \subseteq [k]$ the local decoder can recover all of the corresponding message symbols $x[L],\ldots, x[R]$ while making at most $O(R - L)$ queries to the corrupted codeword. They also observed that the Hadamard code can have amortized locality approaching $1$. 
By contrast, Katz and Trevisan show that codes with locality $\ell < 2$ do not exist unless the  message length is a constant \cite{STOC:KatTre00}.

Given the surprising success in construction ideal amortized \LDC s for Hamming channels, even in relaxed settings, it is natural to wonder whether or not it is possible to obtain similar results for insertions and deletions ({\em Insdel}) channels.

{\em Is it possible to obtain constant amortized locality, constant error tolerant, and constant rate Insdel \LDC s when the channel is oblivious, resource-bounded, or when the sender/receiver share a cryptographic key?}

A construction of an ideal amortized \LDC \ for Insdel channels would be especially surprising since
when compared to known lowerbounds of Hamming \LDC s, the current state of affairs is even worse for Insdel \LDC s.
Blocki et al. \cite{FOCS:BCGLZZ21} proved the first separation result between Insdel and Hamming \LDC s by proving that linear Insdel \LDC s with locality $\ell = 2$ do not exist.
This stands in contrast to Hamming \LDC s, where constructions do exist, albeit with exponential rate (e.g. the Hadamard code).
They additionally provide a exponential rate lowerbound for  any constant locality Insdel \LDC. Their lower bound holds in both traditional and relaxed (private-key and oblivious channel) settings, making progress towards a conjecture raised by Block et al. \cite{FOCS:BCGLZZ21} that constant locality Insdel \LDC s do not exist.
Recently, Gupta resolved this conjecture by showing that constant query deletion codes do not exist by providing a randomized, oblivious adversarial deletion strategy \cite{STOC:Gupta24}.
The results of Gupta extend even to relaxed settings where the sender/receiver share random coins or where the channel is resource bounded i.e., even in relaxed settings it is impossible to construct constant query Insdel \LDC s. Thus, it would be especially surprising to construct an ideal amortized \LDC\ for Insdel channels even in relaxed settings with shared randomness or resource bounded channels. 

On the side of constructions, Ostrovsky and Paskin-Cherniavsky provide the first InsDel \LDC \ via a compilation of a Hamming \LDC s into an Insdel \LDC s \cite{ICITS:OstPas15} with a polylogarithmic blowup to the query complexity and only a constant blowup to the error-tolerance and rate. 
Their construction was revisited by Block et al. \cite{bbgkz20} who reproved and extended their results to construct Locally Correctable Codes (\LCC s), a variant of \LDC s which considers the local decodability of any codeword symbol.
At a high level, their compiler takes a Hamming \LDC \ encoding of the message, partitions it into blocks, and applies an Insdel code to each block.
The local decoder then simulate the underlying Hamming decoder's queries by retrieving and decoding the corresponding Insdel-encoded block to recover the queried Hamming code symbol, where each answered Hamming query incurs a polylogarithmic query blowup.

Unfortunately, a black box approach of instantiating their Insdel \LDC \ constructions with an amortizable Hamming \LDC \ does not preserve the amortized locality. 
Even if all symbols recovered by the decoder are used, the amortized locality will still suffer a polylogarithmic blowup from the Hamming decoder simulation.
Thus, it is unclear whether the Insdel \LDC \ constructions of Ostrovsky and Paskin-Cherniavsky or Block et al. are amortization friendly.
\subsection{Our Contributions}
We provide the first framework for constructing amortized Insdel \LDC s by modifying the compiler of Block et al. \cite{bbgkz20} to be amenable to amortization. 
Our resulting compiler takes in an amortized Hamming \LDC, whose decoder satisfies a properties we define as {\em consecutive interval querying}, and outputs an Insdel amortized \LDC \ with asymptotically equivalent parameters (amortized locality, rate, and error tolerance). 
More specifically, given an amortized Hamming \LDC \ with constant rate, constant error tolerance, and constant amortized locality, whose decoder queries contiguous blocks of polylogarithmic size instead of individual message symbols, our compiler will yield an amortized Insdel \LDC \ with constant rate, constant error tolerance, and constant amortized locality. While prior Insdel \LDC \ compilers suffer from a polylogarithmic blow-up in query complexity of the underlying Hamming \LDC, our amortized compiler preserves amortized locality as well as rate and error tolerance (up to constant factors). Thus, given any  ideal Hamming amortized \LDC\ which satisfies our consecutive interval querying property, we obtain an ideal Insdel amortized \LDC\ achieving constant rate, error tolerance, and amortized locality.
{
\begin{informal}[Informal, see Corollary \ref{corr:compilerNoAssumptions}]
   If there exists ideal (constant rate, constant error tolerance, and constant amortized locality) Hamming amortized \LDC s whose decoder only queries consecutive blocks of size $\Omega(\polylog\ n)$, then there exists ideal Insdel amortized \LDC s (with asymptotically equivalent parameters).
\end{informal}
}

Furthermore, in the private-key setting, we construct the first consecutive interval querying Hamming \LDC\ with constant rate, error tolerance, and constant {\em amortized } locality. We leave it as an open question whether one can construct a consecutive inverval querying Hamming \LDC\ with constant {\em amortized } locality in the information theoretic setting. However, using our updated Hamming-to-Insdel compiler we obtain an ideal amortized Insdel \LDC \ construction (constant rate, constant error tolerance and constant amortized locality) in the private-key setting. Specifically, the local decoder $\Dec^{\tilde{\bY}}(a,b)$ makes at most $O(b-a + \polylog{(n)})$ queries to the corrupted codeword\footnote{$\tilde{\bY}$ is the corrupted codeword output by the PPT-bounded channel who does not have the secret key under the constraint that the edit distance fraction between $\bY$ and $\tilde{\bY}$ is at most $2\delta$. } $\tilde{\bY}$ and (whp) outputs all of the correct symbols $x[a],x[a+1],\ldots, x[b]$ in the interval $[a,b]$. As long as $b-a \geq \polylog(n)$ the amortized locality is $O(1)$ per symbol recovered.
{
\begin{informal}[Informal, see Theorem \ref{thm:b-block-dec-OTpaldc} and Corollary \ref{corr:finalINSDELpaldc}]
   For any $t,n \in \N$ such that $t \mid n$, there exists ideal private-key Hamming amortized \LDC s with codeword length $n$ whose decoder only queries consecutive blocks of size $t$ and has constant amortized locality whenever decoding any consecutive interval of size at least $\omega(t\log n)$.
   This implies that there exists ideal private-key Insdel amortized \LDC s.
\end{informal}
}
We show further that our private-key construction implies an ideal amortized Insdel \LDC\ in resource-bounded settings i.e. when the channel has a bounded amount of some resource (e.g. space, circuit-depth, cumulative memory, etc.).
{
\begin{informal}[Informal, see Corollary \ref{corr:finalINSDELraldc}]
   If there exists ideal private-key Insdel amortized \LDC s, there exists ideal resource-bounded Insdel amortized \LDC s (with asymptotically equivalent parameters).
\end{informal}
}
Our results demonstrate definitively strong separations between locality and amortized locality for Insdel \LDC s. In particular, Gupta \cite{STOC:Gupta24} rules out the existence of {\em any} Insdel LDC with constant locality and error tolerance, even in relaxed settings and even with exponential rate.   
\subsection{Technical Overview}
\paragraph*{Recalling the BBGKZ Compiler}
Block et al. give a construction of a compiler taking any Hamming \LDC\ and compiling it into an Insdel \LDC \cite{bbgkz20}. 
We refer to their construction as the BBGKZ compiler.
In the compiled encoding procedure, the message $\bx$ is first encoded under the Hamming \LDC \ as codeword $\by,$ where $\by$ is then partitioned into size $\tau$ {\em Hamming} blocks $\by = \bw_1 \circ \dots \bw_{B}$. 
Each Hamming block $\bw_j$ along with its index $j$ is encoded using a constant rate Insdel code with the special property stating that the relative Hamming weight of every logarithmic-sized interval is at most $2/5$ (The Schulman-Zuckerman-code \cite{szcode} satisfies this property).
Lastly, each block is pre/postpended with a zero buffer $0^{\phi \tau}$ of length $\phi \tau$, where $\phi \in (0,1),$ resulting in the {\em Insdel} block $\bw'_j$ and the Insdel codeword $\bY = \bw'_1 \circ \dots \circ \bw'_B.$ Intuitively, the zero buffers $0^{\phi \tau}$ help to identify the beginning/end of each block after insertions/deletions.

Then, the compiled Insdel local decoding procedure, when given oracle access to the (possibly corrupted) Insdel codeword $\tilde{\bY}$, attempts to simulate the Hamming decoder with access to a corrupted codeword $\tilde{\by}$ that is close to the Hamming Codeword $\by$ in Hamming Distance. Specifically, when the Hamming decoder queries symbol $i$ of its expected oracle access to (possibly corrupted) codeword $\tilde{\by}$, the Insdel decoding procedure will identify the block $\bw_j$ which contains $\by[i]$, (attempt to) locate the corresponding padded Insdel block $\bw'_j$, undo the padding, run the Insdel decoder to recover the pre-compiled block $\bw_j,$ and return the corresponding symbol $\tilde{\by}[i]$. 
This query simulation process is facilitated by a subroutine we refer to as the \FindBlock \ procedure which (attempts to) locate a particular Insdel block $\bw'_j$ in the presence of corruptions using a noisy binary search procedure.
\paragraph*{Challenge 1: \FindBlock \ is not Amortizable} 
Our first challenge in constructing amortizable Insdel \LDC s is the observation that \FindBlock \ performs asymptotically strictly more queries than the symbols it recovers.
Since the codeword $\bY$ may contain insertions and deletions, in order for \FindBlock \ to locate the Hamming decoder's queries, it performs a {\em noisy binary search} procedure and makes $O(\tau + \tau\polylog (n))$ queries to recover the block corresponding to each query made by the Hamming decoder. The Hamming decoder only uses one of these symbols, but even if the decoder utilized all $\tau$ symbols from the recovered symbols, the amortized locality is still {\em at least} $O(\polylog(n))$ per symbol recovered. Thus, we cannot achieve constant amortized locality by using the BBGKZ compiler as a black box.

Our insight to address this challenge is that \FindBlock\ can be modified to recover a {\em consecutive interval} of blocks rather than just a single block.
We refer to this modified procedure as \FindBlocks, where on input block interval $[a,b] := \{a,\dots,b\}$, \FindBlocks \ will recover pre-compiled blocks $\bw_a,\dots,\bw_b$ and make at most $O((b - a + 1)\tau + \tau \polylog(n))$ queries.
Hence, as long as $b - a + 1 \geq \polylog(n)$, the ratio of recovered symbols to queried symbols will be constant.

\paragraph*{Challenge 2: Proving Correctness of \FindBlocks} 
We show that the modified procedure \FindBlocks \  preserves correctness. 
That is, a majority of the (corrupted) Insdel blocks $\tilde{\bw}'_j$ satisfy the property running the search $\FindBlocks$ with any interval $[a,b]$ that contains $j \in [a,b]$ will yield the original Hamming block $\bw_j$ except with negligible probability. Block et al. \cite{bbgkz20} proved that the procedure $\FindBlock$ will correctly find any {\em local-good block} and that most blocks will be {\em local-good} as long as the number of corruptions is bounded. The main technical challenge results from an inherent mismatch for \FindBlocks \ because an interval $[a,b]$ may contain a mixture of blocks that are (resp. {\em are not}) locally-good. We show that \FindBlocks \ successfully recovers any {\em local-good blocks} in its interval except with negligible probability. 

We prove that correctness in a similar manner to the analysis of Block et al. \cite{bbgkz20}. 
We also utilize the notion of {\em local-good blocks}, where if block $\bw_j'$ is $(\theta,\gamma)$-local-good, it will have at most $\gamma \tau$ corruption with the additional desirable property that any block interval $[a,b]$ surrounding the block, i.e. $j \in [a,b]$, has at least $(1 - \theta)\times|b-a+1|$ blocks that also have at most $\gamma \tau$ corruption. Utilizing analysis from \cite{bbgkz20} we can pick our parameters to ensure that at least $(1-\delta_h)$ fraction of our blocks are $(\theta,\gamma)$-local good. What we prove is that if block $\bw_j'$ is $(\theta,\gamma)$-local good and if $j \in [a,b]$ then calling \FindBlocks \ for the interval $[a,b]$ will recover the original (uncorrupted) $\bw_j$ with high probability along with any other $(\theta,\gamma)$-good block in the interval $[a,b]$. Thus, for any partition $[a_1,b_1],\ldots, [a_z,b_z]$ of $[B]$ calling \FindBlocks \  with each interval $[a_i,b_i]$ would recover the corrupt (uncorrupted) $\bw_j$ for at least $(1-\delta_h)B$ blocks. We will use this observation to reason about the correctness of our decoding procedure. Intuitively, we can view the compiled Insdel decoder as simulating the Hamming decoder with a corrupted codeword whose fractional hamming distance is at most $\delta_h$ from the original codeword.


\paragraph*{Challenge 3: Consecutive Hamming Queries Needed For Amortized Decoding}
The next challenge is utilizing the \FindBlocks \ procedure to obtain amortized locality in the Insdel local decoding procedure. 
Similar to the BBGKZ compiler, our Insdel decoder with oracle access to (possibly corrupted) codeword $\tilde{\bY}$ simulates the compiled Hamming decoder's oracle access to (possibly corrupted) codeword $\tilde{\by}.$ Our hope is that if the underlying Hamming \LDC\ is amortizable then the compiled Insdel \LDC\ will also be amortizable. In particular, we hope to group the Hamming decoder's queries $q_1,\dots,q_\ell$ together into corresponding (disjoint) index intervals $[a_1,b_1],\dots,[a_p,b_p]$ such that (1) $\sum_{i} |b_i+1-a_i| = \Theta(\ell)$, (2) $\{q_1,\ldots, q_\ell\} \subseteq \bigcup_i [a_i,b_i]$, and (3) the average size $\frac{1}{p} \sum_i |b_i+1-a_i|$ of each interval is large enough we can amortize over the calls to $\FindBlocks$ e.g., we need to ensure that the average interval size is at least $\frac{1}{p} \sum_i |b_i+1-a_i| = \Omega(\polylog n)$. Unfortunately, in general, we cannot place any such structural guarantees on the queries made by an (amortizable) Hamming \LDC.

We address this problem by introducing the notion of a {\em $t$-consecutive interval querying} Hamming \LDC decoder which is amortization-friendly under our Insdel compiler. Intuitively, the local decoder for a $t$-consecutive interval querying Hamming \LDC \ accesses the (possibly corrupted) codeword $\tilde{\by}$ by querying for $t$-sized intervals instead of individual message symbols i.e., the decoder may submit the query $[i,i+t-1]$ and the output will be $\tilde{\by}[i,i+t-1]$. In more detail the Hamming decoder takes as input an interval $[L,R]$ of message symbols we would like to recover, makes queries for $t$-sized intervals of the (possibly corrupted) codeword $\tilde{\by}$ and then attempts to output $\bx[L,R]$. Note that if the Hamming decoder makes $\ell/t$ interval queries to $\tilde{\by}$ this still corresponds to locality $\ell$ as the total number of codeword symbols accessed is $t\times (\ell/t) = \ell$.

Correspondingly, when the Hamming decoder is $t$-consecutive interval querying, our Insdel decoder will make at most $2\lceil{\ell / t}\rceil$ calls to \FindBlocks \ on intervals of size $t.$
The resulting query complexity is roughly $\ell \times \left(\frac{\tau \log^3 n}{t} + O(1)\right)$, and so, in other words, the query complexity blow-up is by a factor of $O\left(\frac{\tau \log^3 n }{t}\right)$.
This implies that if there exists a $t$-consecutive interval querying Hamming decoder for $t = \Omega(\tau \log^3 n)$ with constant amortized locality, our compiled Insdel construction will have constant amortized locality. Thus, the primary question is whether or not it is possible to construct $t$-consecutive interval querying Hamming \LDC s with constant amortized locality. We leave this as an open question for worst-case errors, but show that it is possible in settings where the channel is computationally/resource restricted.

\paragraph*{Challenge 4: A $t$-consecutive interval querying \aLDC \ in Private-Key Settings}
We present a $t$-consecutive interval querying private-key amortized Hamming \LDC \ (\paLDC) construction with constant rate, constant amortized locality, and constant error tolerance i.e. ideal parameters.
As a starting point we use the one-time private-key Hamming \LDC \ constructed by Ostrovsky et al. \cite{ICALP:OstPanSah07}.
Blocki and Zhang show that their construction is amortizable and that their construction may be extended to multi-round secret-key reuse with cryptographic building blocks \cite{bloZha25}. 
For simplicity, we focus on the single round construction, which is information-theoretically secure as long as the channel is computationally bounded.

The single round construction first generates the secret key as a randomly drawn permutation $\pi$ and a random string $\br.$ To encode the message $\bx$, it is partitioned into blocks $\bx = \be_1 \circ \dots \be_B$ and each block $\be_j$ is encoded as $\be'_j = \Enc(\be)$ by a constant rate, constant error tolerance Hamming code $\Enc$.
Lastly, the bits of the encoded blocks are permuted by $\pi$ and the random string $\br$ is xor'd i.e. the resulting codeword is $\by = \pi(\be_1' \circ \dots \circ \be'_B) \oplus \br$.
The local decoder can recover any block $\be_j$ by finding the indices of its corresponding encoded block $\be'_j$, reversing the permutation and random masking, and running the Hamming code decoder $\Dec$ to recover $\be_j.$
Unfortunately, the local decoder is not $t$-consecutive interval querying since its queries are not in a contiguous interval. 
The main issue is that the application of the permutation $\pi$ removes any pre-existing block structure in the codeword $\by$ i.e., if the local decoder wants to query for consecutive (pre-permutation) code symbols $(\be_1' \circ \dots \circ \be'_B)[s + 1,s + \ell]$, it would have to query $\by[\pi(s + 1)],\ldots, \by[\pi(s + \ell)]$, which are no longer consecutive with high probability.  

A first attempt to remedy this issue is to permute at the block level instead of the bit level i.e. we draw permutation $\pi$ over $[B]$ and set the codeword instead as $\by = (\be'_{\pi(1)} \circ \dots \circ \be'_{\pi(B)}) \oplus \br.$
While this admits contiguous access when recovering any block $\be_j$ it defeats the original purpose of the permutation. The permutation $\pi$ (and the one-time pad $\br$) ensure that (whp)  an adversarial channel cannot produce too many corruptions in any individual block. If we permute at the block level then it is trivial for an attacker to corrupt entire blocks. In particular, the adversary could focus their entire error budget to corrupt symbols at the start of the codeword. This simple strategy will always corrupt a constant fraction of blocks $\be'_j.$ Thus, the proposed code will not satisfy correct decoding. 

Our primary insight is that we can apply the permutation at an intermediate ``sub-block'' granularity to ensure correctness and simultaneously support consecutive interval queries. In particular, each encoded block $\be'_j,$ for $j \in [B],$ is further partitioned into $\beta$ sub-blocks $\be'_j = (\bepsilon_{(j - 1)\beta + 1} \circ \dots \circ \bepsilon_{j\beta})$ and the permutation is applied over sub-blocks $\bepsilon_r$, for $r \in [B\beta]$, rather than the Hamming blocks $\be_j'$ themselves i.e. the codeword is \[\by = \left( \bepsilon_{\pi(1)} \circ \cdots \circ \bepsilon_{\pi(\beta)} \circ \cdots \cdots \circ \bepsilon_{\pi(B\beta)}\right) \oplus \br.\]  
\begin{figure}
    \centering
    \[
\begin{array}{ccccc}
    \bx = & \be_1 & \circ \dots \circ & \be_B &  \\[6pt] 
        &  \ \ \  \downarrow_{\Enc}     &  \ \ \  \downarrow_{\Enc} &  \ \ \  \downarrow_{\Enc} & \\[6pt] 
        & \be'_1 & \circ \dots \circ & \be'_B & \\[6pt]  
        &   \shortparallel    & \shortparallel &   \shortparallel    &  \\[6pt]  
         &(\bepsilon_1 \circ \dots \circ \bepsilon_\beta) & \circ \dots \circ & (\bepsilon_{(B - 1)\beta + 1} \circ \dots \circ  \bepsilon_{B\beta}) & \\[6pt]
         &   \ \ \ \ \ \ \ \ \downarrow_{\pi(\cdot) \oplus \br}    & \ \ \ \ \ \ \ \ \downarrow_{\pi(\cdot) \oplus \br} &   \ \ \ \ \ \ \ \downarrow_{\pi(\cdot) \oplus \br}    &  \\[6pt]
    \by=  & \resizebox{8pt}{!}{(}\bepsilon_{\pi(1)} \circ \dots \circ \bepsilon_{\pi(\beta)} & \circ \dots \circ &  \bepsilon_{\pi\left((B - 1)\beta + 1\right)} \circ \dots  \circ \bepsilon_{\pi(B\beta)}\resizebox{8pt}{!}{)} & \oplus \ \br
\end{array}
\]
    \caption{An overview of our private-key Hamming \aLDC\ decoder, satisfying correct decoding (of any $\be_j$ Hamming block) and consecutive interval querying (of any $\bepsilon_{\pi(r)}$ sub-block) simultaneously.}
    \label{fig:enc-diagram}
\end{figure}
Intuitively, while a constant fraction of the sub-blocks $\bepsilon_{r}$ can still be corrupted with non-negligible probability, as long as the parameter $\beta$ is suitably large (e.g., $\beta = O(\polylog n)$), we can ensure that, except with negligible probability, the overall error in each encoded block $\be'_j$ will still be at most a $\delta$-fraction. Thus, except with negligible probability we can recover $\be_j$ (the uncorrupted content in block $j$) by making $\beta$ consecutive interval queries to obtain $\bepsilon_{r}$ for each $r \in [(j - 1)\beta +1,j\beta]$.  This is exactly what we needed to achieve amortization with our Ideal Insdel compiler. 

The above construction assumes that the sender and receiver share a secret key and only allows the sender to transmit one encoded message. However, if the channel is probabilistic polynomial time (PPT) then the construction can be extended allow for the sender to transmit multiple messages using standard cryptographic tools e.g., pseudo-random functions\cite{ICALP:OstPanSah07,bloZha25}. Furthermore, if we assume stronger resource bounds on the channel (space, sequential depth etc...) then one can use cryptographic puzzles to completely eliminate the requirement for a secret key using ideas from prior works e.g., \cite{SCN:AmeBloBlo22, ITC:BloKulZho20,ISIT:BloBlo21}. 

\paragraph*{Putting it All Together: Ideal Insdel \aLDC s in Relaxed Settings}
With the building blocks established prior, the $\aLDC$ Hamming-to-Insdel compiler and the consecutive interval querying, ideal Hamming $\paLDC$, we construct the first ideal Insdel \aLDC \ within relaxed settings (private-key and resource-bounded).
Naturally, we first construct a private-key, ideal Insdel \aLDC \ by directly using our compiler on the $t$-consecutive interval querying, ideal Hamming $\paLDC$. 
For compiler block size $\tau$, our construction achieves constant amortized locality by setting $t = \Omega(\tau \log^3 n)$ as discussed previously. 

The remaining challenge lies in showing that the compiler's correctness is preserved in private-key settings. 
Intuitively, we argue via a reduction to the security of the underlying Hamming \paLDC. 
Suppose there exists a probabilistic polynomial time adversary $\calA$ who introduces an Insdel corruption pattern into compiled codeword $\bY$ of their chosen message $\bx$ that causes a decoding error with non-negligible probability. 
Then, we can construct a probabilistic polynomial time adversary $\calB$ who introduces a Hamming error pattern into (pre-compiled) codeword $\by$ by (1) using our compiler procedure to transform Hamming codeword $\by$ into the compiled Insdel codeword $\bY$, (2) Giving compiled Insdel codeword $\bY$ to adversary $\calA$ and receiving corrupted Insdel codeword $\tilde{\bY}$, and finally, (3) constructing the pre-compiled codeword $\tilde{\by}$ by the same simulation process done by the compiled decoder via the \FindBlocks \ procedure. 
To streamline the reduction argument, we draw inspiration from Block and Blocki \cite{ISIT:BloBlo21} and repackage our compiler into a convenient form for our reduction. 
All together, we have an Ideal Insdel \paLDC. 

Next, we generalize the \paLDC\ construction to construct an \aLDC\ for resource-bounded settings (\raldc).
Similarly, our first step is constructing a $t$-consecutive interval querying Hamming \raldc. 
To do this, we make use of the Hamming private-key-to-resource-bounded \LDC \ compiler by Ameri et al. \cite{SCN:AmeBloBlo22}, which was observed by Blocki and Zhang to be amortizable with only a constant blow-up to the amortized locality \cite{bloZha25}. 
At a high level, their compiler makes use of a cryptographic primitive known as a {\em cryptographic puzzle} to embed the secret key within the codeword. 

More specifically, cryptographic puzzles consists of two algorithms $\PuzzGen$ and $\PuzzSolve.$
On input seed $\bc$, $\PuzzGen$ outputs a puzzle $\bZ$ with solution is $\bc$ i.e. $\PuzzSolve(\bZ) = \bc.$ 
The security requirement states that adversary $\calA$ in a defined algorithm class $\R,$ cannot solve the puzzle $\bZ$ and cannot even distinguish between tuples $(\bZ_0,\bc_0,\bc_1)$ and $(\bZ_1,\bc_0,\bc_1)$ for random $\bc_i$ and $\bZ_i = \PuzzGen(\bc_i).$

Then, the Hamming $\raldc$ compiler takes a $\paLDC$ $(\p \Gen, \p \Enc, \p \Dec)$, and on on message $\bx$ will (1) generate the secret key $\sk \leftarrow \p \Gen(1^\lambda; \bc)$ with some random coins $\bc$ (2) Compute the \paLDC \ encoding of the message $\p \by \leftarrow \Enc_{\psk}(\bx)$ (3) Compute an encoding $\by_*$ 
of the puzzle $\bZ$ using a code with low locality that recovers the entire message (known as an $\LDC^*$) and (4) output $\by = \p \by \circ \by_*$.
Intuitively, a resource-bounded adversary will not be able to solve the puzzle, while the $\raldc$ decoder $\Dec$ will have sufficient resources to solve the puzzle, recover the secret key $\sk$, and run the amortized local decoder $\Dec_{\psk}$.

We make an additional observation that if the compiler is instantiated with a $t$-consecutive interval querying $\paLDC,$ then the resulting $\raldc$ will also be $t$-consecutive interval querying. 
Thus, when applied to our amortized Hamming-to-Insdel compiler with $t = \Omega(\tau \log^3 n),$ we acheive constant amortized locality. 
Lastly, the same reduction argument as in the ideal Insdel \paLDC \ construction may be used, with an additional subtle detail of allowing sufficient resources for the reduction.
That is, for the class of resource-bounded channels/adversaries $\R$ that our Insdel \aLDC \ is secure against, the underlying Hamming \aLDC \ must be secure against a wider class of adversaries $\overline{\R}$, where $\overline{\R}$ include adversaries $\calA \in \R$ whom can additionally compute the simulation process in the private-key reduction.

\subsection{Preliminaries}
\paragraph*{Notation} Let $x \leftarrow y$ represent a variable assignment of $x$ as $y$ and $x \unif X$ denote a uniform random assignment of $x$ from space $X$.
Let $\circ$ denote the concatenation function.
For $l \leq r,$ $[l,r] := \{l,l+1,\dots,r\}$ denotes an interval from $l$ to $r$ inclusive of its ends points. 
Similarly, $[l,r) := \{l,l+1,\dots, r -1\}$ denotes an interval from $l$ to $r$ exclusive of the right endpoint.
Over an alphabet $\Sigma$, bolded notations $\bx$ denote vectors/words while non-bolded $x \in \Sigma$ denote single symbols.
For a word $\bx$, $x[i]$ denotes the symbol of $\bx$ at index i, while $\bx[l,r] := (x[l],x[l + 1],\dots,x[r])$ denotes the subword of $\bx$ projected to indices in interval $[l,r].$
It will also be useful to define a short-hand for retrieving symbols from ``blocks" of a word $\bx$: $\bx\langle i\rangle_b := \bx[(i - 1)b + 1, ib]$ for any block size $b$ and index $i.$

For any words $\bx,\by \in \Sigma^n$, $\ham(\bx,\by)$ denotes the hamming distance between $\bx$ and $\by$ i.e. $|\{i \in [n]  : x_i \neq y_i \}|$.
Further, for $\bz \in \Sigma^*, \ed(\bx,\bz)$ denotes the edit distance between $\bx$ and $\bz$ i.e. the number of symbol insertions and deletions to transform $\bx$ into $\bz.$
We say $\bx$ and $\by$ (resp. $\bx$ and $\bz$) are $\delta$-hamming-close (resp. $\delta$-edit-close) when $\ham(\bx,\by) \leq \delta |x|$ (resp. $\ed(\bx,\bz) \leq 2\delta|\bx|$).

We recall the formal definitions of an \LDC \ and an amortized \LDC.
\begin{definition}[\LDC]
    A $(n,k)$-code $\calC = (\Enc,\Dec)$ (over $\Sigma$) is a $(\ell,\delta,\eps)$-locally decodable code $(\LDC)$  for hamming errors (resp. insDel errors) if for every $\bx \in \Sigma^k,$ $\tilde \by \in \Sigma^*$ such that $\tilde{\by}$ is $\delta$-hamming-close (resp. $\delta$-edit-close) to $\Enc(\bx),$ and every index $i \in [k],$ we have 
    \begin{align*}
    \Pr\lbrack\msf \Dec^{\tilde \by}(i) &= x[i] \rbrack \geq 1 - \eps,
    \end{align*}
and $\Dec^{\tilde \by}$ makes at most $\ell$ queries to $\tilde \by$.  
\end{definition}
\begin{definition}[\aLDC \ \cite{bloZha25}]\label{def:aLDC}
    A $(n,k)$-code $\calC = (\Enc,\Dec)$ (over $\Sigma$) is a $(\alpha, \kappa, \delta, \epsilon)$-amortizeable LDC (\aLDC)  for hamming errors (resp. insDel errors) if for every $\bx \in \Sigma^k,$ $\tilde \by \in \Sigma^*$ such that $\tilde{\by}$ is $\delta$-hamming-close (resp. $\delta$-edit-close) to $\Enc(\bx),$ and every interval   $[L,R] \subseteq [k]$ or size $R-L+1 \geq \kappa$  we have 
    \begin{align*}
    \Pr\lbrack\msf \Dec^{\tilde \by}(L,R) &= (x_i : i \in [L,R]) \rbrack \geq 1 - \eps,
    \end{align*}
and $\Dec^{\tilde \by}$ makes at most $\alpha \times (R - L + 1) $ queries to $\tilde \by$. 
\end{definition}

We remark that Blocki and Zhang $\cite{bloZha25}$ also suggested a more general version of Definition \ref{def:aLDC} where the interval $[L,R]$ can be replaced with an arbitrary subset $Q \subseteq [k]$ of of size $|Q| \geq \kappa$. However, they only constructed ideal Hamming \aLDC s under Definition \ref{def:aLDC} i.e., with intervals instead of arbitrary subsets. Our compiler transforms any consecutive interval querying Hamming \aLDC\ to a Insdel \aLDC. While the compiler itself would still work with respect to the more general definition, it seems less likely that one could construct consecutive interval querying Hamming \aLDC s with ideal rate/locality/error tolerance even in relaxed settings when amortizing over arbitrary subsets $Q \subseteq [k]$. We leave this as an open question for future research.

Throughout the paper, it will be assumed codes are over the binary alphabet i.e. ${\Sigma = \{0,1\}}$ unless specified otherwise.
The primary metrics of interest for Hamming/Insdel \aLDC s are the rate $k/n,$ amortized locality $\alpha$, and the error/edit tolerance $\delta.$

\subsection{Organization}
We start with reintroducing and modifying the encoding and decoding procedures of the BBGKZ compiler in Section \ref{sect:bbgkz-modified}. 
The modification will take various steps organized into the following subsections:
in Subsection \ref{subsec:goodblocks}, the notion of good blocks and intervals are introduced to analyze our new \FindBlocks \ procedure in Subsection \ref{subsect:findblocks}. 
We then present the overall construction of our Hamming-to-Insdel \aLDC \ compiler in Subsection \ref{subsect:aldcDecoder}.

In Section \ref{sect:comp-bounded}, we present our construction of an ideal Insdel \aLDC s in settings where the encoder/decoder share a secret-key (\paLDC).
More specifically, this construction first relies on a modified Hamming \paLDC \ construction in Subsection $\ref{subsect:privateIdealHamming}$, where the local decoder performs consecutive interval queries. 
The ideal Hamming \paLDC \ is then compiled into our ideal Insdel \paLDC \ in Subsection \ref{subsect:privateIdealInsdel}.
Lastly, in Section \ref{sect:resource-bounded}, we present  our construction of an ideal Insdel \aLDC s in settings where the channel is resource-bounded (\raldc).
Section \ref{sect:resource-bounded} follows a symmetric structure to the previous section: we present a consecutive interval querying Hamming \raldc \ construction in Subsection \ref{subsect:resourceIdealHamming}, which is then compiled into an ideal Insdel \raldc \ construction in Subsection \ref{subsect:resourceIdealInsdel}.

\section{The Hamming \aLDC \ to Insdel \aLDC \ Compiler}\label{sect:bbgkz-modified}
Our results crucially rely on a modified BBGKZ compiler suited for amortized local decoding. 
We start with the encoding procedure.
Interestingly, we will not need to make any changes to the encoding procedure of the BBGKZ compiler, which takes as parameters the block size $\tau \in \N$ and padding rate $\phi \in (0,1)$: the message $\bx$ will first be encoded into a codeword $\by$ using a Hamming \LDC.
Next, the codeword $\by$ is broken up into blocks of equal size $\tau$ as $\by = \bw_1 \circ \dots \bw_{B}$, where each block and its index $(i \circ \bw_i)$ is 1) encoded with a specifically chosen Insdel code 2) pre-and-post-pended with $\phi \tau$ many $0$s.
The result is the overall codeword $\bY.$
Formally, we describe the encoder $\Enc$ as first computing the Hamming codeword $\by$, then applying a Hamming-to-Insdel compiler $\EncCompile$ to transform it into an Insdel codeword $\bY.$
The encoder $\Enc$ and compiler \EncCompile \ are formally described below.

Let $\h \calC = (\h \Enc, \h \Dec)$ be a Hamming \LDC.
Let $\insdel \calC = (\insdel \Enc, \insdel \Dec)$ be an Insdel binary code.
\begin{weirdFrame}{$\EncCompile_{\phi}(\by)$}
{\setlength{\parindent}{0cm}\textbf{Input}: Hamming codeword $\by \in \{0,1\}^m,$ padding rate $\phi \in (0,1).$}

{\setlength{\parindent}{0cm}\textbf{Output}: Insdel codeword $\bY \in \{0,1\}^n.$}
    \begin{enumerate}
        \item Parse $\by = \bw_1 \circ \dots \circ \bw_{|\by | / \tau}$ where $\bw_j \in \{0,1\}^{\tau}$ for all $j \in [ \ |\by| / \tau \ ]$.
        \item Compute $\bw'_j \leftarrow 0^{\phi\tau} \circ \insdel \Enc(j \circ \bw_j) \circ 0^{\phi\tau}$ for all $j \in [ \ |\by| / \tau \ ]$.
        \item Output $\bY =  \bw'_1  \circ \dots \circ  \bw'_{|\by| / \tau}.$
    \end{enumerate}
\end{weirdFrame}
\begin{weirdoFrame}{$\Enc(\bx)$}
\textbf{Parameters}: Block size $\tau \in \N$, Padding rate $\phi \in (0,1).$

{\setlength{\parindent}{0cm}\textbf{Input}: message $\bx \in \{0,1\}^k.$}

{\setlength{\parindent}{0cm}\textbf{Output}: Insdel codeword $\bY \in \{0,1\}^n.$}
    \begin{enumerate}
        \item Compute $\by \leftarrow \h \Enc (\bx).$
        \item Output $\bY \leftarrow \EncCompile(\by).$
    \end{enumerate}
\end{weirdoFrame}
Our main modifications to the BBGKZ compiler will be in the construction of the decoder $\Dec.$
Recall that the original BBGKZ decoder simulated oracle access to the Hamming \LDC \ decoder $\h \Dec$ by recovering the Hamming blocks corresponding to the requested queries of $\h \Dec.$
Since insertions and deletions in the corrupted codeword $\tilde{\bY}$ modify the structure of the blocks, this simulation is underpinned by a {\em noisy binary search} procedure we denote as the \FindBlock \ procedure. 
More specifically, the \FindBlock \ procedure on input block index $j$ will iteratively cut an initial search radius $[1,\tilde{n}]$ by a fractional amount, until the search radius is a constant size larger than the block size $\tau.$
To decide how to cut the search radius on each iteration, the procedure samples $N = \theta(\log^2 n)$ blocks indices by decoding $N$ noisy $\bw'_{j_1},\dots \bw'_{j_b}$ blocks within the middle of the search radius to $(j_b \circ \bw_{j_b})$. 
On the block indices $j_1,\dots,j_b,$
their median $j_{\msf{med}}$ is compared to the desired block index $j:$ depending if $j < j_{\msf{med}}$, \FindBlock \ will cut off a fraction of the front or the back of the search radius.

Since the search radius decreases by a constant fraction each iteration, there are $O(\log n)$ iterations, where on each iteration, $\theta(log^2 n)$ blocks are read of $O(\tau)$ size.
Thus, the query complexity to answer a single query of $\h \Dec$ is $O(\tau \log^3 n)$.

Unfortunately, amortization is not feasible with the current Hamming decoder simulation method. 
Intuitively, since the \FindBlock \ procedure must query a multiplicative $\polylog \ n$ factor to recovers $\tau$ symbols, the amortized locality will be at least $\Omega(\polylog \ n).$

Our solution will  recover multiple blocks at once using the observation that the \FindBlocks \  procedure can be extended from recovering a single block to an interval of contiguous blocks.
That is, we will modify the noisy binary search procedure to take in a block interval $a \leq b \in [B]$, and we change the stopping condition to be linear to the interval size $b - a + 1.$
Further, the consideration on how to shrink the search radius will take into account the respective endpoints $a$ and $b$ of the interval rather than a single block $j.$

While this change is relatively straight forward on paper, the analysis for the noisy binary search procedure BBGKZ compiler does not immediately follow for our modified procedure \FindBlocks.
Thus, to prove correctness for the modified \FindBlocks \ procedure, we reintroduce the good block analysis from Block et al. \cite{bbgkz20}, and modify their proof structure to take into account the recovery of an interval of blocks instead of a single block.

\subsection{Good Blocks and Intervals}\label{subsec:goodblocks}
In this subsection, we introduce good block notation to analyze the modifications we will make to the BBGKZ decoder and its sub-routine \FindBlock.

\paragraph*{Additional Notation}First, we fix notation consistent with the encoder definition above for the rest of the paper.
We will refer to the Hamming encoding $\by$ in step $1$ as the {\em Hamming encoding} and the final codeword $\bY$ as the {\em Insdel encoding} or simply as the {\em encoding.}
Let $\bx \in \{0,1\}^k$ be the message of size $k$ and let $\by = \h \Enc(\bx) \in \{0,1\}^m$ be the Hamming encoding with size $m$.
Let $B = m / \tau$ be the number of blocks.
Define $\beta$ such that the rate of the overall encoding is $1 / \beta$.
Similarly, define $\insdel \beta$ such that the rate of the Insdel code $\insdel \calC$ is $1 / \insdel \beta$.
Then, the rate of the Hamming encoding is $(\insdel \beta \log B + 2 \phi) / \beta.$ 
Let $\h \delta$ and $\insdel \delta$ be the error tolerances of the Hamming encoding and the Insdel code respectively.
Let each $\bw_j = \by\langle j \rangle_\tau \in \{0,1\}^\tau$ denote {\em Hamming block} $j$.
Similarly, each $\bw'_j \in \{0,1\}^{\beta \tau}$ denotes the {\em Insdel block/ block} $j$.
Then, the overall codeword is $\bY \in \{0,1\}^n$ has length $n = \beta m.$
We consider any corrupted codeword $\tilde{\bY} \in \{0,1\}^{\tilde{n}}$ such that $\ed(\bY,\tilde{\bY}) \leq 2\delta n.$

We define the {\em block decomposition} as a mapping from codeword symbols to their blocks post corruption.
\begin{definition}\label{def:blockDecomposition}
    A block decomposition of $\tilde{\bY}$ is a non-decreasing mapping $\phi: [\tilde{n}] \rightarrow [B].$
\end{definition}
By definition, the pre-image of each block decomposition defines a partition of $[\tilde{n}]$ into $B$ intervals $\{\phi^{-1}(j): j \in [B]\}.$
Each of these intervals define the block structure, where the j'th interval defines the j'th block in the corrupted codeword.
In other words, for each $j \in [B]$, $\phi^{-1}(j)$ are all the indices  corresponding to block $j.$

As in the BBGKZ compiler, our modified decoder will rely on a noisy recovery process, where the decoder will need to locate blocks without knowing the block boundaries that have been modified by insdel errors.
For any search interval $[l,r) \subseteq [\tilde{n}],$ it will be helpful to consider its corresponding block interval i.e. the smallest interval $[L,R - 1] \subseteq [B]$ such that $[l,r) \subseteq [\tau L, \tau (R- 1)].$
Block intervals are defined formally below.
\begin{definition}[Block Interval]\label{def:blockInterval}
    The block interval of an interval $I = [l,r)  \subseteq [\tilde{n}]$ is defined as $\bigcup_{i = l}^{r-1} \phi^{-1}(\phi(i)) \subseteq [\tilde{n}]$. 
    An interval $I$ is a block interval if the block interval of $I$ is itself. 
    Equivalently, every block interval has the form $\BI[a,b] := \bigcup_{i = a}^b \phi^{-1}(j)$ for some $a,b \in [B].$
\end{definition}

We say a block is {\em good} if it does not have too many edit corruptions.
Intuitively, good blocks will correspond to Insdel blocks that are correctly decodable, except with negligible probability.
\begin{definition}[$\gamma$-Good Block]\label{def:goodBlock}
    For $\gamma \in (0,1)$ and $j \in [B]$, a block $j$ is $\gamma$-good (with respect to $\Tilde{\bY}$) if 
    \[\ed(\Tilde{\bY}[\phi^{-1}(j)],\bw_j') \leq \gamma \tau.\]
    Otherwise, block $j$ is $\gamma$-bad.
\end{definition}
{\em Good blocks} may be naturally extended to {\em good intervals}, where the summed edit corruptions of each block within the interval is not too much. 
Additionally, we constrain the number of bad blocks in any good interval.
\begin{definition}[$(\theta,\gamma)$-Good Interval]\label{def:goodInterval}
    A block interval $\BI[a,b]$ is $(\theta,\gamma)$-good (with respect to $\Tilde{\bY}$) if the following hold:
    \begin{enumerate}
        \item $\sum_{j = a}^b \ed\left(\Tilde{\bY}[\phi^{-1}(j)],\bw'_j \right) \leq \gamma \tau(b - a + 1).$
        \item The number of $\gamma$-bad blocks in the block interval $\BI[a,b]$ is at most $\theta \times (b - a +1)$.
    \end{enumerate}
    Otherwise, the interval is $(\theta,\gamma)$-bad.
\end{definition}
Lastly, we define the notion of a locally good block, which captures the idea that any interval including such a block must also be good.
\begin{definition}[$(\theta,\gamma)$-Local Good Block]\label{def:localGoodBlock}
    For $\theta,\gamma \in (0,1),$ block $j \in [B]$ is $(\theta,\gamma)$-local good (with respect to $\Tilde{\bY}$) if for every $a,b \in [B]$ such that $a \leq j \leq b,$ the interval $\BI[a,b]$ is $(\theta,\gamma)$-good.
    Otherwise, block $j$ is $(\theta,\gamma)$-locally bad.
\end{definition}
Block et al. derive useful properties about the size changes and total number of good and bad blocks respectively \cite{bbgkz20}. 
Their observations are summarized in Lemma \ref{lem:goodBlockProperties} below.
\begin{lemma}[\cite{bbgkz20}]\label{lem:goodBlockProperties}
The following all hold:
    \begin{enumerate}
        \item For all $\gamma$-good blocks $j,$ $\size{\phi^{-1}(j)} \in [(\beta - \phi \gamma) \tau,(\beta + \phi \gamma) \tau]$.
        \item For all $(\theta,\gamma)$-good block intervals $\BI[a,b],$ $\size{\BI[a,b]} \in [(b-a + 1)(\beta - \phi\gamma), (b-a + 1)(\beta + \phi\gamma)]$.
        \item The total fraction of $\gamma$-bad blocks is at most $\frac{2\delta\beta}{\gamma\phi}$.
        \item The total fraction of $(\theta,\gamma)$-local-bad blocks is at most $(4/ \gamma\phi) (1 + 1/\theta)\delta \beta.$
    \end{enumerate}
\end{lemma}
\subsection{Extending \FindBlock \ to \FindBlocks} \label{subsect:findblocks}
We extend the \FindBlock \ procedure to recover multiple Hamming blocks in an amortization-friendly way.

We start with recalling specific details of the \FindBlock \ procedure.
The \FindBlock \ procedure for recovering a single block takes as input the desired block index $j \in [B]$ and the initial search radius $[1,\tilde{n}].$ 
The procedure will iteratively split the search radius into three contiguous, linear-sized segments, a beginning, a middle, and an end segment, where it will choose to either cut the beginning or the end segment.
To decide which, it randomly samples $N = \theta(\log^2 n)$ indices $i_1,\ldots, i_N$ within the middle segment and recover their respective block indices $j_1,\dots,j_N$.
More specifically, for each index $i_b,$ a subroutine \BlockDec \ is called which queries a $O(\tau)$ interval around $i_b$ to decode the noisy block containing $(j \circ \bw_j)$.
The median $j_{\msf{med}}$ of these retrieved indices is then calculated, and, depending if $j < j_{\msf{med}}$ or $j \geq j_{\msf{med}}$, 
cuts either the end or beginning segment from the search interval respectively.
The binary search continues until the search radius is linear in the size of the block for a sufficiently small constant.

Lastly, an interval decoding procedure is then run over the symbols in the final search interval, where it is guaranteed to recover any local-good block located within the interval.
The subroutines used, \BlockDec \ and the final interval decoding, are based on local and global variations of the SZ-code decoding algorithm, finding the pre/post-pended $0's$ on the Insdel blocks to locate the blocks in between.

\newcommand{\SimFindBlock}{\ensuremath{\texttt{Sim-RecoverBlock}}}
Intuitively, we may modify \FindBlock \ to recover an interval $[a,b] \subseteq [B]$ of contiguous Hamming blocks instead of a single block by stopping the search once its search radius is linear to the {\em size of the interval} versus the size of a single block.
{Lastly, we run a procedure we call $\SimFindBlock$ to recover each block $j \in [a,b]$ within the final search radius $[l,r]$. If we naively call  $\FindBlock^{\tilde{\bY}}(l,r,q)$ repeatedly for every block $q \in [a,b]$ then each call to  $\FindBlock$ will execute its own noisy binary search procedure in the interval $[l,r]$ resulting in redundant queries which would blow up our amortized locality i.e., the total number of queries would be too large i.e. $\Omega((l-r) \polylog (l-r))$. Instead, we simply make $l-r+1$ queries to $\tilde{Y}$ to recover the entire substring $\tilde{Y}[l,r]$ and then run $ \SimFindBlock(\tilde{\bY}[l,r], l,r,q)$ for each $q \in [a,b]$ where  $\SimFindBlock(\tilde{\bY}[l,r], l,r,q)$ uses the hardcoded string $\tilde{Y}[l,r]$ to simulate the execution of $\FindBlock^{\tilde{\bY}}(l,r,q)$. 
}
We call this overall extended procedure \FindBlocks, and present the construction below (modifications to the \FindBlock \ procedure are in \textcolor{MidnightBlue}{blue}).
\begin{weirdFrame}{$\FindBlocks^{\tilde{\bY}}(l,r,a,b)$}
    {\setlength{\parindent}{0cm}\textbf{Parameters}: Block size $\tau \in \N, \gamma \in (0,1)$}
    
    {\setlength{\parindent}{0cm}\textbf{Input}: Search range $l < r \in [\tilde{n} + 1],$ \new{block range $a \leq b \in [B]$}}

    {\setlength{\parindent}{0cm}\textbf{Output}: pre-compiled blocks $\bs \in \{0,1\}^{\new{(b - a + 1)} \tau}$}

    {\setlength{\parindent}{0cm}Let $c = 36 (\beta - \gamma)$.}   

    {\setlength{\parindent}{0cm}Let $\rho = \min \left \{ \frac{1}{4} \times \frac{\beta - \gamma}{\beta + \gamma}, 1 - \frac{3}{4} \times \frac{\beta + \gamma}{\beta - \gamma}\right\}$.}  
    
    \begin{enumerate}
        \item While $r - l > c\new{(b -a +1)}\tau$:
        \begin{enumerate}
            \item Let $m_1 \leftarrow (1 - \rho) l + \rho r,$ and $m_2 \leftarrow \rho l + (1 - \rho) r.$
            \item For $p = 1,\dots,N = \theta (\log^2n)$, 
            \begin{enumerate}
                \item randomly sample $i_p \unif [m_1,m_2)$ 
                \item set $S \leftarrow \BlockDec^{\tilde{\bY}}(i_p)$
                \item If $S = \bot,$ ignore and continue 
                \item Else, parse $(j_p,\bw_{j_p}) \leftarrow S$
            \end{enumerate}
            \item Let $j_{med} \leftarrow $ median$(j_1,\dots,j_N)$ (ignore $\bot$) 
            \item \new{If $a > j_{med}$, set $l = m_1.$}
            \new{Otherwise, set $r = m_2$.}
        \end{enumerate}
        \item \new{ Make $(l-r+1)$ queries to recover $\tilde{\bY}[l,r]$}.
        \item \new{Output $\left [ \SimFindBlock(\tilde{\bY}[l,r], l,r,q) : q \in [a,b] \right ]$}.
        
    \end{enumerate}
\end{weirdFrame}


The first step is to prove that \FindBlocks \ on input $(1,\tilde{n} + 1,a,b)$ for any block interval $[a,b] \subseteq [B]$ will return all Hamming blocks $\bw_j$ within that interval, given that the blocks are good.
We will then argue that by a careful setting of parameters, the number of correctly recovered blocks is sufficient for an overall correct decoding process.

To start, we recall theorems proved by Block et al. on the correctness of the key procedures \BlockDec \ and \FindBlock\ when applied over good/local-good blocks.
 {We remark that in the original BBGKZ compiler, \BlockDec\ on codeword index $i$ outputs the entire corresponding Hamming block encoding of $j \circ \bw_j$. 
In \FindBlocks, we have \BlockDec\ just return the block index $j$.}
\begin{lemma}[\cite{bbgkz20}, Theorem 20.1]\label{thm:20.1}
     {There exists $c_1 < 1$ and $c_2 \geq 1$ such that} for $\delta = c_1\insdel \delta$, $\tau = c_2\log n$ and $\gamma$-good block $j$,
    \[\Pr_{i \in \phi^{-1}(j)}\left[\BlockDec^{\Tilde{\bY}}(i) \neq j\right] \leq \gamma,\]
    Additionally, $\BlockDec$ makes $O(\tau)$ queries to $\Tilde{\bY}.$
\end{lemma}
\begin{lemma}[\cite{bbgkz20}, Theorem 19]\label{thm:20.2}
    Let $[l,r) \subseteq [\tilde n]$ be any interval.
    Let $S_j$ be the random variable corresponding to $ \SimFindBlock(\tilde{\bY}[l,r], l,r,q)$.
    For any $j \in [B]$ such that $\phi^{-1}(j) \subset [l,r]$ and block $j$ is $(\theta,\gamma)$-local-good, then \[\Pr[S_j = \bw_j] \geq 1 - \negl(\lambda).\]
    Additionally, the computation of $\bs$ makes at most $O(c \times (b - a + 1))$ queries to $\tilde{\bY}$. 
\end{lemma}
{We remark that the original statement of Lemma \ref{thm:20.2} pertains to the output of $\FindBlock^{\tilde{\bY}}(1,\tilde n,j)$, i.e., the output of $\FindBlock$ when ran over the entire codeword interval $[1,\tilde n]$. 
However, one can observe that the proof extends to any general search interval $[l,r) \subseteq [\tilde n]$ as long as $\phi^{-1}(j) \subset [l,r].$
Hence, we may state the equivalent result for \SimFindBlock \ in Lemma \ref{thm:20.2}.}

We derive the query complexity of \FindBlocks \ based on the query complexities of \BlockDec\ and \SimFindBlock\ stated in the introduced Lemmas \ref{thm:20.1} and \ref{thm:20.2}.

\begin{lemma}\label{lem:findBlocksQueryComplexity}
    For any word $\tilde{\bY} \in \{0,1\}^{\tilde{n}}$, interval $[a,b] \subseteq [B]$, and constants $\beta,\gamma > 0,$ on input $(1,\tilde{n} + 1,a,b),$ $\FindBlocks$ has query complexity $O((b - a)\tau + \tau\log^3n)$.
\end{lemma}
\begin{proof}
    The search radius at each iteration is cut by a $\rho \geq 1/4$ fraction at each iteration. 
    Thus, there are at most $O(\log n)$ iterations.
    On each iteration, $N = \theta(\log^2 n)$ calls to the $\BlockDec$ procedure are made, where each call makes $O(\tau)$ queries by Lemma \ref{thm:20.1}.
    Lastly, we run \SimFindBlock \ on an interval of length at most $c(b - a + 1) \tau$ resulting in at most $c(b - a + 1) \tau$ queries. By Lemma \ref{thm:20.2}, the overall query complexity is  $O((b - a)\tau + \tau\log^3n)$.
\end{proof}

The \FindBlocks \ procedure may be called to retrieve the Hamming blocks of any interval. We will guarantee that if the corresponding blocks are local-good, they are recovered except with negligible probability.
\begin{restatable}{theorem}{findblocks}\label{thm:findBlocks}
    There exists constants $\gamma,\theta \in (0,1)$  {$, c_1 < 1$ and $c_2 \geq 1$} such that for $\delta = c_1\insdel \delta$, ${\tau = c_2 \log n}$, and any interval $[a,b] \subseteq [B],$ if there is a $j \in [a,b]$ such that block $j$ is $(\theta,\gamma)$-local-good, $\FindBlocks^{\tilde{\bY}}(1,\tilde{n} + 1, a,b)\langle j - a + 1\rangle_\tau = \bw_j$ except with $\negl(n)$ probability.
\end{restatable}
The proof of Theorem \ref{thm:findBlocks} will rely on several auxiliary lemmas. 
We reintroduce the follow two lemmas from the prior work of Block et al. {Recall that the rate of our constructed code is $1 / \beta.$}
\begin{lemma}[\cite{bbgkz20}, Lemma 23]\label{lem:probBlockDecOverGoodinterval}
    Let $[l,r)$ be an interval with block interval $\BI[L,R]$.
    Suppose $\BI[L,R - 1]$ is a $(\theta,\gamma)$-good interval. Then, 
    \[\Pr_{i \in [l,r)}[\BlockDec^{\tilde{\bY}}(i) \neq \phi(i)] \leq \gamma + \theta + \frac{\gamma}{\beta}.\]
\end{lemma}
{Lemma \ref{lem:probBlockDecOverGoodinterval} states that applying \BlockDec \ over a random chosen block $\tilde{\bw}_j$ in a $(\theta,\gamma)$-good interval $\BI[L,R - 1]$ has an increased failure probability of at most $\theta + \frac{\gamma}{\beta}$ over the failure probability of applying \BlockDec \ over a single $\gamma$-good block (which is at most $\gamma$ by Lemma \ref{thm:20.1}).
For carefully chosen $\theta$ and $\gamma$, this difference will be at most a constant fraction smaller.}

\begin{lemma}[\cite{bbgkz20}, Lemma 24]\label{lem:goodBlockToGoodinterval}
    For any interval $[a,b] \subseteq [B],$ let $l,r \in [\tilde{n}]$ be indices such that $r - l \geq 18 (\beta + \gamma) \tau \times (b - a + 1).$
    Let $\BI[L,R - 1]$ be the block interval for $[l,r).$
    Let $\rho = \min \left \{ \frac{1}{4} \times \frac{\beta - \gamma}{\beta + \gamma}, 1 - \frac{3}{4} \times \frac{\beta + \gamma}{\beta - \gamma}\right\}$.
    For $m_1  = (1 - \rho) l + \rho r$ and $m_2 = \rho l + (1 - \rho) r$, let $\BI[M_1,M_2 -1]$ be the block interval of $[m_1,m_2).$
    Suppose for some $L \leq c \leq M_1,$ block $c$ is $(\theta,\gamma)$-local good. 
    If $\frac{\beta + \gamma}{\beta - \gamma} < 4/3$, then
    \begin{enumerate}
        \item $M_1 \leq L + (R - L) / 3$, and $M_2 \geq L + 2(R - L) / 3.$
        \item $\BI[M_1,M_2 - 1]$ is a $(2\theta,2\gamma)$-good interval.
    \end{enumerate}
\end{lemma}
{At each iteration in the \FindBlocks\ procedure, we sample $N = \theta(\log^2 n)$ random codeword indices $i_1,\dots,i_N$ from  ``middle" interval $[m_1,m_2)$ within a surrounding search interval $[l,r) = [l,m_1) \cup [m_1, m_2) \cup [m_2,r)$ to decide how to cut the search interval.
We compute the corresponding block indices $j_1,\dots,j_N$ using the \BlockDec\ procedure to derive an estimate $j_{med}$ on the median block index.
Lemma \ref{lem:goodBlockToGoodinterval} states that (1) the block interval $\BI[M_1,M_2-1]$ corresponding to codeword interval $[m_1,m_2)$ retains at least a constant fraction of blocks present in the surrounding block interval $\BI[L,R-1]$ (corresponding to codeword interval $[l,r)$), and (2) within the ``left" block interval $\BI[L,M_1]$, if there is some $(\theta,\gamma)$-local good block $c \in [L,M_1],$ then the middle block interval $[M_1,M_2 - 1]$ will be a $(2\theta,2\gamma)$-good interval.}

{
We use Lemma \ref{lem:goodBlockToGoodinterval} in conjunction with Lemma \ref{lem:probBlockDecOverGoodinterval} to show that with except with negligible probability, \FindBlocks\ procedure correctly cuts the search interval $[l,r)$, either setting $l = m_1$ or $r = m_2$, based on the computed median block index $j_{med}$. 
Specifically, we argue that at least half block indices $j_1,\dots,j_N$ computed from sampled codeword indices $i_1,\dots,i_N$ will be computed correctly, except with negligible probability. 
Thus, the median $j_{med}$ is computed correctly except with negligible probability, and we cut the search interval correctly.
This is formally stated in the following Lemma \ref{lem:iterateFindBlocks}.
}

\begin{lemma}\label{lem:iterateFindBlocks}
    Let $[a,b] \subseteq [B]$ be any interval such that there is an index $c \in [a,b]$ such that block $c$ is a $(\theta,\gamma)$-local good block. 
    Denote by $l^{(t)},r^{(t)}$ the values of $l,r$ at the beginning of the t-th iteration iteration of the while loop in $\FindBlocks^{\tilde{\bY}}(1,\tilde{n} + 1,a ,b)$.

    If $N \in \theta(\log^2 n), r^{(t)} - l^{(t)} \geq 36 (\beta + \gamma) \tau \times (b - a + 1),$ and $2 ((1 + 1/\beta)\gamma + \theta) < 1/3,$
    \[\Pr\left[\bigcup_{j \in [a,b]} \phi^{-1}(j) \subseteq \left[l^{(t + 1)},r^{(t + 1)}\right) \middle| \bigcup_{j \in [a,b]}\phi^{-1}(j) \subseteq \left[l^{(t)},r^{(t)}\right)\right] \geq 1 - \negl(n)\]
\end{lemma}
\begin{proof}
Let $\BI[L,R - 1]$ be the block interval for interval $[l^{(t)},r^{(t)})$.
As in $\FindBlocks,$ let 
\begin{align*}
    m_1 &= (1 - \rho) l^{(t)} + \rho r^{(t)}, \\
    m_2 &= \rho l^{(t)} + (1 - \rho) r^{(t)},
\end{align*}
for $\rho = \min \left \{ \frac{1}{4} \times \frac{\beta - \gamma}{\beta + \gamma}, 1 - \frac{3}{4} \times \frac{\beta + \gamma}{\beta - \gamma}\right\}$.
Let $\BI[M_1, M_2 - 1]$ be the block interval of $[m_1,m_2).$
Observe that when $\bigcup_{j \in [a,b]} \phi^{-1}(j) \subseteq [m_1,m_2),$ we have that $\bigcup_{j \in [a,b]} \phi^{-1}(j) \subseteq [l^{(t +1)},r^{(t + 1)})$ by construction. 
Thus, assume that $\bigcup_{j \in [a,b]} \phi^{-1}(j) \not\subseteq [m_1,m_2)$, implying that either $L \leq a \leq M_1$ or $M_2 - 1 \leq b \leq R.$

Suppose first that $L \leq a \leq M_1$ and $b < M_2.$
We will show that the Lemma holds for this case (and its symmetric case, when $M_2 - 1 \leq b \leq R$ and $a < L)$, and we will finish by arguing that the case when $L \leq a \leq M_1$ and $M_2 - 1 \leq b \leq R$ also holds.

By the assumption that $r^{(t)} - l^{(t)} \geq 36 (\beta + \gamma) \tau \times (b - a +1)$, we have that $m_2 - m_1 \geq (r^{(t)} - l^{(t)}) / 2 \geq 18(\beta + \gamma)\tau \times (b - a +1)$.
Thus, we can apply Lemma \ref{lem:goodBlockToGoodinterval} to have that 
\begin{itemize}
    \item $M_2 - M_1 \geq (R - L) / 3 \geq 12$, and 
    \item $\BI[M_1,M_2 - 1]$ is a $(2\theta,2\gamma)$-good interval. 
\end{itemize}
Since $\BI[M_1,M_2 - 1]$ is the block interval of $[m_1,m_2),$ by Lemma \ref{lem:probBlockDecOverGoodinterval}, we have
\[\Pr_{i \in [m_1,m_2)}[\BlockDec^{\tilde{Y}}(i) \neq \phi(i)] \leq 2\gamma + 2 \theta + \frac{2 \gamma}{\beta} < \frac{1}{3}\]
Let $i_1,\dots,i_N$ be the sampled indices in \FindBlocks\ on iteration $t$.
Let $X_p$ be the indicator random variable for event $\{\BlockDec(i_p) = \bot\} \cup \{\BlockDec(i_p) < a\}$.
Then, $\bigcup_{j \in [a,b]} \phi^{-1}(j) \not\subseteq [l^{(t + 1)},r^{(t + 1)})$ occurs if and only if $\sum_{p = 1}^N X_p \geq N / 2$ i.e. the median of the sampled block indices is less than $a.$

We upperbound this by introducing another indicator random variable.
Let $Y_b$ be the indicator random variable for event ${\{\BlockDec(i_b) \neq \phi(i_b)\}}$. 
Observe that any sampled index ${i_b \in [m_1,m_2) \subseteq \BI[M_1,M_2 -1 ]}$ implies that ${\phi(i_b) \geq M_1 \geq a.}$
Thus, we have that event $\{\BlockDec(i_b) = \phi(i_b)\}$ implies event $\{\BlockDec(i_b) \geq a\}$, so $Y_b \geq X_b$. 
Since $\E[Y_b] < \frac{1}{3}$ and $N = \theta(\log^2 n),$ we can apply a Chernoff bound to conclude
\begin{align*}
    \Pr\left[\sum_{j = 1}^N X_j \geq \frac{N}{2}\right] &\leq \Pr\left[\sum_{j = 1}^N Y_j \geq \frac{N}{2}\right]  \\
    &\leq \Pr\left[\sum_{j = 1}^N Y_j \geq \left(1 + \frac{1}{2}\right) \sum_{j = 1}^N\E[Y_j]\right] \leq \exp\left(-\frac{N}{36}\right) \leq \negl(n).
\end{align*}
In the case that $M_2 - 1 \leq b \leq R$ and $a < L$, a symmetric argument can be made. 
It is left to show that the lemma holds when $L \leq a \leq M_1$ and $M_2 - 1 \leq b \leq R$.
However, this implies $m_2 - m_1 \leq (b - a + 1)\tau$, and since we have $r - l \geq 36(\beta - \gamma) \tau\times (b - a + 1) \geq 36 (b - a + 1)\tau$ by the loop condition, this case cannot happen since $m_2 - m_1 \geq 18 (b - a + 1)\tau$.
\end{proof}
Prior to proving Theorem \ref{thm:findBlocks}, we need Lemma \ref{thm:20.1}, \ref{lem:probBlockDecOverGoodinterval}, and \ref{lem:iterateFindBlocks} to hold simultaneously. Each of these Lemmas induce overlapping constraints, so we need to show there exists a valid parameter setting.
Block et al. show that these parameters constraints are feasible. 
Note that while Lemma \ref{lem:iterateFindBlocks} is modified from their original analysis, the constraint remains the same.
\begin{theorem}[ \cite{bbgkz20}, Proposition 22]\label{thm:constraints}
    There exists constants $\gamma,\theta \in (0,1)$, $\delta \in \Omega(\insdel \delta)$, and $\tau \in \Omega(\log n)$ such that the preconditions of Lemma \ref{thm:20.1}, \ref{lem:probBlockDecOverGoodinterval}, and \ref{lem:iterateFindBlocks} hold simultaneously.
\end{theorem}
We are now ready to prove Theorem \ref{thm:findBlocks}.
\findblocks*
\begin{proof}
    Set $\gamma, \theta, \delta,$ and $\tau$ according to Theorem \ref{thm:constraints}.
    There are $T = O(\log n)$ iterations until $r - l \leq c\tau \times (b - a + 1)$ for $c = 36(\beta + \gamma)\tau.$
    As before, denote $l^{(t)},r^{(t)}$ as the values of $l,r$ at the start of the $t$-th iteration of the \FindBlocks\ loop. 
    
    Then, we have for any $j \in [a,b]$ such that block $j$ is a $(\theta,\gamma)$-local good block,
    \begin{align*}
        \Pr\left[\bs_j= \bw_j\right] &\geq \Pr\left[\phi^{-1}(j) \subseteq [l^{(T)},r^{(T)}]\right] \times \Pr\left[ \bs_j = \bw_j  \ \middle | \ \phi^{-1}(j) \subseteq [l^{(T)},r^{(T)}] \right].
    \end{align*}
    Starting with the first term,
    \begin{align*}
        \Pr\left[\phi^{-1}(j) \subseteq [l^{(T)},r^{(T)}]\right] &= \prod_{t = 1}^{T - 1} \Pr\left[\phi^{-1}(j) \subseteq [l^{(t + 1)},r^{(t + 1)}] \ \middle | \ \phi^{-1}(j) \subseteq [l^{(t)},r^{(t)}] \right] \\
        &\geq (1 - \negl(n))^T \tag{Lemma \ref{lem:iterateFindBlocks}} \\
        &\geq 1 - T\times \negl(n) = 1 - \negl(n).
    \end{align*}
    For the second term, by Lemma \ref{thm:20.2}, we have $\Pr\left[ \bs_j = \bw_j  \ \middle | \ \phi^{-1}(j) \subseteq [l^{(T)},r^{(T)}] \right] \geq 1 - \negl(n)$, so we can conclude
    \[ \Pr\left[\bs_j= \bw_j\right] \geq (1 - \negl(n))^2 = 1 - \negl(n).\]    
\end{proof}
By Theorem \ref{thm:findBlocks}, a call to \FindBlocks \ on a local-good interval will successfully return the Hamming block except with negligible probability. 
We extend Theorem \ref{thm:findBlocks} to show that over any interval partition $\calP = \{[a_i,b_i] : 1 \leq a_i \leq b_i < a_{i +1} \leq B \}$ of $[B]$, the total number of Hamming blocks that are incorrectly recovered is bounded by the Hamming error tolerance $\h \delta$ for a careful setting of the overal error tolerance $\delta.$
\begin{theorem}\label{thm:FindBlocksSimulates}
    Let $\calP = \{[a_i,b_i] : 1 \leq a_i \leq b_i < a_{i +1} \leq B \}$ be any partition of $[B].$ 
    Then, let $\bs^{[a_i,b_i]}_j$ be the random variable defined as the output of $\FindBlocks^{\tilde{\bY}}(1,\tilde{n} + 1, a_i, b_i)\langle j - a_i +1 \rangle_\tau$. 
    For parameters set by Theorem \ref{thm:constraints}, if $\delta = \frac{\h \delta \phi \gamma}{8 \beta (1 + 1 / \theta)}$, then 
    \[\Pr \left[\sum_{[a_i,b_i] \in \calP}\sum_{j = 1}^{b_i - a_i +1} \mathbb{1}\left(\bs^{[a_i,b_i]}_j \neq \bw_j\right) \geq \h \delta B\right] \leq \negl(n),\]
    where the probability is taken over the joint distribution $\left\{\bs^{[a_i,b_i]}_j : [a_i,b_i] \in \calP, j \in [b_i - a_i + 1]\right\}.$
\end{theorem}
\begin{proof}
    Fix a partition  $\calP = \{[a_i,b_i] : 1 \leq a_i \leq b_i < a_{i + 1} \leq B\}$ of $[B].$
    Let $\Good \subset [B]$ be the set of $(\theta,\gamma)$-local-good blocks. 
    Let $\Good^{[a,b]}$ be the subset of $(\theta,\gamma)$-local-good blocks in interval $[a,b].$
    Then, by Lemma $\ref{lem:goodBlockProperties}$,
    \[\overline{\Good} \leq \left(1 + \frac{1}{\theta}\right)\frac{\delta \beta B}{\phi \gamma} = \frac{\delta B}{2}.\]
    Thus,
    \begin{align*}
        \Pr \left[\sum_{[a_i,b_i] \in \calP}\sum_{j = 1}^{b_i - a_i +1} \mathbb{1}\left(\bs^{[a_i,b_i]}_j \neq \bw_j\right) \geq \h \delta B\right] &= \Pr\left[\sum_{\substack{[a_i,b_i] \in \calP \\ j \in \Good^{[a_i,b_i]}}} \mathbb{1}\left(\cdot\right) + \sum_{j \in \overline{\Good}} \mathbb{1}\left(\cdot \right) \geq \h \delta B\right] \\
        &\leq \Pr\left[\sum_{\substack{[a_i,b_i] \in \calP \\ j \in \Good^{[a_i,b_i]}}} \mathbb{1}\left(\bs^{[a_i,b_i]}_j \neq \bw_j\right) \geq \frac{\h \delta B}{2}\right] \\
        &= \Pr\left[\sum_{\substack{[a_i,b_i] \in \calP \\ j \in \Good^{[a_i,b_i]}}} \mathbb{1}\left(\bs^{[a_i,b_i]}_j \neq \bw_j\right) \geq \frac{\out \delta B}{2}\right] \leq \negl(n),
    \end{align*}
    where the last line holds by Theorem \ref{thm:findBlocks}.
\end{proof}

\subsection{The Amortized Insdel Decoder} \label{subsect:aldcDecoder}
We construct an amortized Insdel local decoder $\Dec$ utilizing \FindBlocks\ to simulate oracle access for the underlying Hamming decoder $\h \Dec$.  
Our Insdel decoder $\Dec$ follows a similar procedure to the BBGKZ compiler decoder. 
$\Dec$ calls the underlying Hamming decoder $\h \Dec$, and on its queries $i_1,\dots,i_q \in [m]$ to the (possibly corrupted) Hamming codeword $\tilde{\by}$, will respond by calling \FindBlocks \ to recover the corresponding Hamming blocks and queried symbols.
Correctness of our procedure will follow from Theorem $\ref{thm:FindBlocksSimulates}$, which roughly states that the symbols returned by \FindBlocks \ are equivalent to what $\h \Dec$ expects from true oracle access to the (possibly corrupted) Hamming codeword $\tilde{\by}$.

Our next step is to ensure that the queries made by our Insdel decoder amortize over the \FindBlocks \ calls by restricting our Hamming decoder to only make contiguous queries, where each contiguous interval is of a fixed size.
Suppose the Hamming decoder $\h \Dec$ is guaranteed to make queries in contiguous intervals of size $t$, say $[u_1 + 1,u_1 + t],\dots,[u_p + 1,u_p + t] \subseteq [m]$ for some $t.$ 
Intuitively, when the Insdel decoder $\Dec$ is simulating the Hamming decoder's queries with \FindBlocks, it can then make $O(p)$ calls to $\FindBlocks$ to recover codeword intervals of size $O(t)$.
For large enough $t$, the number of \FindBlocks\ calls $O(p)$ decreases and the Insdel decoder's queries will amortize.

We formally describe a Hamming decoder with this desirable property $t$-consecutive interval querying Hamming decoder. 
\begin{definition}[Consecutive Interval Querying]
    Consider any $(m,k)$-code $\calC = (\Enc, \Dec)$ that is a $(\alpha,\kappa,\delta,\eps)$-$\aLDC$.
    For any word $\tilde{\by}$, interval $[L,R] \subseteq [k]$ with $R - L + 1 \geq \kappa$, and random coins $r,$ let $\msf{Query}(\tilde{\by},[L,R],r) := \{i_1,\dots,i_q\} \subseteq [\tilde{n}]$ denote the codeword indices queries when $\Dec^{\tilde{\by}}(L,R)$ is ran with randomness $r$. 
    Code $\calC$ and decoder $\Dec$ are $t$-{\em consecutive interval querying} if the set $\msf{Query}(\tilde{\by},[L,R],r)$ can be partitioned into $q / t$ disjoint intervals $[u_1,v_1],\ldots, [u_{q/t},v_{q/t}]$ of size $t$ i.e.,
    \begin{enumerate}
        \item  $v_j - u_j+1 = t$ for all $j \leq q/t$,
        \item $\msf{Query}(y,[L,R],r) = \bigcup_{j \leq q/b} [u_j,v_j] $, and
        \item $[u_{j_1},v_{j_1}] \cap [u_{j_2},v_{j_2}] = \emptyset$ for all $j_1 \neq j_2$. 
    \end{enumerate} 
\end{definition}
We now show how to construct the amortized local decoder $\Dec$ using \FindBlocks\ and a $t$-consecutive interval querying Hamming $\aLDC$ decoder $\h \Dec$ for a value of $t$ to be specified in the analysis.
We choose the SZ-code as the Insdel code $\insdel \calC$, which has constant rate ${1/\insdel \beta = \Omega(1)}$ and constant error-tolerance $\insdel \delta = \theta(1)$.
\begin{construction}\label{constr:modifiedBBGKZ}
Let $\h \calC = (\h \Enc, \h \Dec)$ be a $t$-consecutive interval querying Hamming \LDC \ such that $\tau$ divides $t$.
    Let $\insdel \calC = (\insdel \Enc, \insdel \Dec)$ be an Insdel binary code.
\begin{weirdFrame}{$\EncCompile_{\phi}(\by)$}
{\setlength{\parindent}{0cm}\textbf{Input}: Hamming codeword $\by \in \{0,1\}^m,$ padding rate $\phi \in (0,1).$}

{\setlength{\parindent}{0cm}\textbf{Output}: Insdel codeword $\bY \in \{0,1\}^n.$}
    \begin{enumerate}
        \item Parse $\by = \bw_1 \circ \dots \circ \bw_{|\by | / \tau}$ where $\bw_j \in \{0,1\}^{\tau}$ for all $j \in [ \ |\by| / \tau \ ]$.
        \item Compute $\bw'_j \leftarrow 0^{\phi\tau} \circ \insdel \Enc(j \circ \bw_j) \circ 0^{\phi\tau}$ for all $j \in [ \ |\by| / \tau \ ]$.
        \item Output $\bY =  \bw'_1  \circ \dots \circ  \bw'_{|\by| / \tau}.$
    \end{enumerate}
\end{weirdFrame}
\begin{weirdoFrame}{$\Enc(\bx)$}
\textbf{Parameters}: Block size $\tau \in \N$, Padding rate $\phi \in (0,1).$

{\setlength{\parindent}{0cm}\textbf{Input}: message $\bx \in \{0,1\}^k.$}

{\setlength{\parindent}{0cm}\textbf{Output}: Insdel codeword $\bY \in \{0,1\}^n.$}
    \begin{enumerate}
        \item Compute $\by \leftarrow \h \Enc (\bx).$
        \item Output $\bY \leftarrow \EncCompile(\by).$
    \end{enumerate}
\end{weirdoFrame}
    \begin{weirdoFrame}{$\Dec^{\tilde{\bY}}(L,R)$}
\textbf{Parameters}: Block size $\tau \in \N$, Padding rate $\phi \in (0,1).$

{\setlength{\parindent}{0cm}\textbf{Input}: $1 \leq L \leq R \leq \tilde{n}$}

{\setlength{\parindent}{0cm}\textbf{Output}: word $\tilde{\bs} \in \{0,1\}^{(R - L + 1)\tau}$}
    \begin{enumerate}
        \item Suppose $\h \Dec(L,R)$ queries disjoint intervals $[u_1,v_1],\dots,[u_{p},v_{p}] \subseteq [m]$ of size $t.$  
        \item For each $r \in [p],$
        \begin{enumerate}
            \item Compute $j \in \left[ \ \lfloor n / t \rfloor \ \right]$ such that $[u_r,v_r] \subset [(j-1)t + 1, (j + 1)t].$
            \item Let $t' = \lceil t/ \tau \rceil$ and compute 
            \begin{align*}
                \bs^{(j - 1)} &\leftarrow \RecoverBlocks^{\tilde{\bY}}(1,\tilde{n}  + 1, (j-1)t' + 1, jt') \\
                \bs^{(j)} &\leftarrow \RecoverBlocks^{\tilde{\bY}}(1,\tilde{n}  + 1, jt' + 1, (j + 1)t').
            \end{align*}
        \end{enumerate}
         \item From $\left\{\bs^{(j)}\right\}$, send $\h \Dec$ the bits corresponding to intervals $[u_1,v_1],\dots,[u_{p},v_{p}]$.
        \item Output the output of $\h \Dec. $
    \end{enumerate}
\end{weirdoFrame}
\end{construction}
Our Insdel decoder $\Dec$ processes the $t$-sized intervals $[u_1,v_1],\dots,[u_p,v_p] \subseteq [m]$ queried by the Hamming decoder $\h \Dec$ into a corresponding $t' = \lceil t / \tau\rceil$-sized block interval inputs for \FindBlocks \ (step 2.a).
Note that for ease of presentation, the decoder always computes \FindBlocks\ on two $t'$ sized block intervals (step 2.b). 
The amortized locality can be optimized by a constant factor by calling \FindBlocks \ on the exact block intervals queried by the Hamming decoder $\h \Dec.$
However, it will be convenient to assume the Insdel decoder \Dec \ calls \FindBlocks \ in a predictable manner for all inputs 
when proving correctness i.e. the input interval to \FindBlocks\ is always $((j - 1)t' +1,jt')$ for some $j$, and asymptotically, there is no change to the amortized locality.
\begin{theorem}\label{thm:compiler}
    Let $(m,k)$-code $\h \calC = (\h \Enc, \h \Dec)$ be a $t$-consecutive interval querying Hamming $(\h \alpha, \h \kappa, \h \delta, \h \eps)$-\aLDC.
    Then, for any block size $\tau \in \Omega(\log n)$ such that $\tau \mid n$, $\calC = (\Enc,\Dec)$ in Construction \ref{constr:modifiedBBGKZ} is a $\left(\alpha, \kappa,\delta, \eps\right)$-$\aLDC$ for $\alpha = O\left(\frac{\h \alpha \tau \log^3n}{t} \right), \kappa = \h \kappa, \delta = \Omega(\h \delta)$, and $\eps = \h \eps + \negl(n).$
\end{theorem}
\begin{proof}
Suppose $\h \Dec^{\tilde{\by}}(L,R)$ queries disjoint intervals ${[u_1,v_1],\dots,[u_{p},v_p] \subset [m]}$.
For each $r \in [p]$, we compute \FindBlocks \ on block intervals ${[(j - 1)t' + 1,jt']}$ and ${[jt' + 1,(j + 1)t']}$ such that
$[u_r,v_r] \subset [(j-1)t + 1, (j + 1)t]$.
Let \[\calP = \left\{[1,t'],[t'+1, 2t'],\dots,[B - t' + 1, B]\right \}\] be a partition of $[B]$ in equal $t'$-sized intervals.
   Let $\bs^{(j)}$ be the random variable defined as the output of $\FindBlocks^{\tilde{\bY}}(1,\tilde{n} + 1, (j - 1) \times t' + 1, j \times t')$. 
    Define $\tilde{\by}$ as the random string defined by $\tilde{\by}\langle j  \rangle_{t'} = \bs^{(j)}$ for each $j \in [ \lfloor n / t' \rfloor]$.
    Then, since $\Ham(\tilde{\by}_u,\by) \leq \h \delta m$ is implied by the event 
    \[\sum_{[(j-1)t' + 1,j t'] \in \calP}\sum_{r = 1}^{t} \mathbb{1}\left(\bs^{(j)}_r \neq \bw_j\right) \leq \h \delta B,\]
    by Theorem \ref{thm:FindBlocksSimulates}, $\Pr[\Ham(\tilde{\by},\by) \leq \h \delta m] \geq 1 - \negl(n)$.
    Thus, from the view of $\h \Dec,$ it is interacting with $\tilde{\by}$ over partition $\calP$, and so by definition, for any $R - L + 1 \geq \h \kappa$,
    \[\Pr\left[\h \Dec^{\tilde{\by}}(L,R) = (x_i : i \in [L,R]) \ \middle | \ \Ham(\tilde{\by},\by) \leq \h \delta m\right] \geq 1 - \eps_u.\]
    By construction of $\Dec,$ for any $R - L + 1 \geq \h \kappa,$
    \begin{align*}
        \Pr&\left[\Dec^{\tilde{\bY}}(L,R) = (x_i : i \in [L,R]  \right] \\
        &\geq \Pr[\Ham(\tilde{\by}_u,\by) \leq \h \delta m] \Pr\left[\h \Dec^{\tilde{\by}_u}(L,R) = (x_i : i \in [L,R])  \ \middle | \ \Ham(\tilde{\by}_u,\by) \leq \h \delta m \right] \\ 
        &\geq (1 - \negl(n)) \times \left(1 - \h \eps\right) \\
        &\geq 1 - \h \eps - \negl(n). 
    \end{align*}
    Since $\h \calC$ is $t$-consecutive interval querying which makes at most $\h \alpha (R - L + 1)$ queries, the decoder $\Dec$ calls $\FindBlocks$ at most $\frac{2\h \alpha (R - L + 1)}{t}$ times on block intervals of size $t' = \lceil t / \tau \rceil$. 
    By Lemma \ref{lem:findBlocksQueryComplexity}, the query complexity of $\Dec$ is 
    \[\alpha(R - L + 1) = \frac{2\h \alpha (R - L +1)}{t} \times O\left((\lceil t / \tau \rceil \times \tau) + \tau \log^3n\right) \text{ and so } \alpha = O\left(\frac{\h \alpha \tau \log^3n}{t} \right).\qedhere\]
\end{proof}
If a $t$-consecutive interval querying, ideal Hamming \aLDC\ does exists, we obtain the following corollary.
\begin{corollary}\label{corr:compilerNoAssumptions}
If there exists a $t$-consecutive interval querying, ideal Hamming \aLDC \ for ${t = \Omega(\tau \log^3n)}$, then there exists an ideal Insdel $\aLDC$.
\end{corollary}
We construct $t$-consecutive interval querying, ideal Hamming $\aLDC$s in private-key and resource-bounded settings and leave a construction for worst-case Hamming errors as an open question. 

\section{Ideal Insdel \aLDC s in Private-key Settings}\label{sect:comp-bounded}

Given the compiler and resulting Theorem \ref{thm:compiler} in the previous section converting an ideal Hamming \aLDC \ $\h \calC$ to ideal Insdel \aLDC\ $\calC$ whenever $\h \calC$ is $\Omega(\tau\log^3 n )$-consecutive interval querying, we show that such an \aLDC \ exist with private-key assumptions.
We start by recalling the definition of a  private-key \aLDC \ (\paLDC).

\begin{definition}[\paLDC \ \cite{bloZha25}]\label{def:paLDC}
    Let $\lambda$ be the security parameter. A triple of probabilistic polynomial time algorithms $(\Gen,\Enc,\Dec)$ over $\Sigma$ is a private $(\alpha, \kappa, \delta, \epsilon,q)$-amortizeable LDC (\paldc) for Hamming errors (resp. Insdel errors) if 
    \begin{itemize}
        \item for all keys $\sk \in \text{Range}(\Gen(1^\lambda))$ the pair $(\Enc_{\sk},\Dec_{\sk})$ is a $(n,k)$-code $\calC$ over $\Sigma$, and
        \item for all probabilistic polynomial time algorithms $\calA$ there is a negligible function $\mu$ such that
        \[\Pr[\paldcgame(\calA,\lambda,\alpha,\kappa,\delta,\alpha,q) = 1] \leq \mu(\lambda),\]
        where the probability is taken over the randomness of $\calA, \Gen,$ and $\paldcgame$. 
        The experiment \paldcgame\ is defined as follows:
    \end{itemize}
     \begin{weirdFrame}
        {\paldcgame$(\calA,\lambda,\alpha,\kappa,\delta,q)$}
        The challenger generates secret key $\msf{sk} \leftarrow \Gen(1^\lambda)$. 
        For $q$ rounds, on iteration $h$, the challenger and adversary $\calA$ interact as follows:
        \begin{enumerate}
            \item The adversary $\calA$ chooses a message $\bx^{(h)} \in \Sigma^k$ and sends it to the challenger.
            \item The challenger sends $\by^{(h)} \leftarrow \Enc_{\sk}(\bx^{(h)})$ to the adversary.
            \item The adversary outputs $\Tilde{\by}^{(h)} \in \Sigma^*$ that is $\delta$-Hamming-close (resp. $\delta$-edit-close) distance to $\by^{(h)}$.
            \item If there exists $L^{(h)}, R^{(h)} \in [k]$ such that $R^{(h)} - L^{(h)} + 1 \geq \kappa$ and
            \[\hspace{-\leftmargin}\Pr\left[\Dec^{\Tilde{\by}^{(h)}}_\sk(L^{(h)},R^{(h)}) \neq \bx[L^{(h)},R^{(h)}]\right] > \eps(\lambda)\] s.t. $\Dec^{\Tilde{\by}^{(h)}}_\sk(.)$ makes at most $(R - L + 1) \alpha$ queries to $\tilde \by$, then this experiment outputs $1$.
        \end{enumerate}
        If the experiment did not output $1$ on any iteration $h$, then output $0.$
    \end{weirdFrame}
\end{definition}
Note that for $\paLDC$s, we assume that the local decoder takes in a consecutive interval $[L,R]$ rather than an arbitrary subset $Q$ as input. 
Blocki and Zhang show that explicit, ideal Hamming \paLDC \ constructions exist for such decoders, and the existence of ideal \paLDC s whose local decoders take in arbitrary subsets $Q$ is left as an open problem.
\subsection{A consecutive interval querying, Ideal Hamming \paLDC}\label{subsect:privateIdealHamming}
We construct a $t$-consecutive interval querying, ideal Hamming \paLDC \ by modifying of the private-key \LDC \ construction by Ostrovsky et al. \cite{ICALP:OstPanSah07}.
Recall that the encoding procedure on input message $\bx$ first partitions it into equal-sized blocks $\bx = \be_1 \circ \dots \be_B$, for all $j \in [B]$. Each block is then individually encoded as $\be'_j = \Enc(\be_j)$ for each $j \in [B]$ by a constant rate, constant error tolerant, and constant alphabet size Hamming code encoding $\Enc$ (e.g. the Justesen code) to form the encoded message $\bx' = \be'_1 \circ \dots \circ \be'_B$.
Lastly, an additional secret-key permutation $\pi$ and random mask $\br$ are applied i.e.  the codeword is $\by = \pi(\bx') \oplus \br$, which effectively makes the codeword look random from the view of a probabilistic polynomial time channel.

Blocki and Zhang \cite{bloZha25} observe that the encoding procedure by Ostrovsky et al. is amenable to amortization when the recovered symbols are in a consecutive interval $[L,R].$
Their amortized local decoder, with access to the secret key, undoes the permutation $\pi$ and random mask $\br$, recovers the contiguous blocks $\be'_{s + 1},\dots,\be'_{s + \ell}$ corresponding to requested $[L,R],$ decodes each block, and returns the corresponding queried symbols.
Unfortunately, since we applied a permutation $\pi$ in the encoding procedure, the amortized local decoder does not make consecutive interval queries to the (corrupted) codeword $\tilde{\by}$.

To add $t$-consecutive interval querying to the current decoder, we modify the encoding scheme permutation step to permute $t$-sized {\em sub-blocks} of the blocks $\be'_j$ in encoded message $\bx' = \be'_1 \circ \dots \circ \be'_B$, rather than permuting individual bits.
We highlight in \new{blue} significant changes made from the original $\paLDC$ presented by Blocki and Zhang.
\begin{construction}\label{constr:paLDCv1BlockDecoding} 
Suppose that $\calC = (\Enc_{\calC},\Dec_{\calC})$ is a Hamming $(A,a)$-code over an alphabet $\Sigma$ with rate $R$.
Let $c = \log |\Sigma|$. 
Define $\p \calC = (\p \Gen, \p \Enc, \p \Dec)$ with message length $k$, codeword length $m = \frac{k}{R}$, block size $cA$, and sub-block size $t$ (dividing $cA$) as follows:
\begin{weirdFrame}{$\Gen(1^\lambda) $}
\begin{enumerate}
    \item Generate $\br \unif \{0,1\}^{m}$ and uniformly random permutation \new{$\pi : [m / t] \rightarrow [m / t]$.}
    \item Output $\sk \leftarrow (\br,\pi).$
\end{enumerate}
\end{weirdFrame}
\begin{weirdFrame}{$\Enc_{\sk}(\bx)$}

{\setlength{\parindent}{0cm}\textbf{Input}: Message $\bx \in \{0,1\}^k$}

{\setlength{\parindent}{0cm}\textbf{Output}: Codeword $\by \in \{0,1\}^{m}$}
\begin{enumerate}
    \item Parse $(\br,\pi) \leftarrow \sk$.
    \item Let $\bx \leftarrow \be_1 \circ \be_2 \circ \dots \circ \be_{B}$ where each $\be_s \in \Sigma^{a}$ (and $B = k / ca$). 
    \item For each $s \in [B]$:
    \begin{enumerate}
        \item set $\be_s' \leftarrow \Enc_{\calC} (\be_s)$, and 
        \item let \new{$\be_s' = \bepsilon_{(s - 1)\beta + 1} \circ \dots \circ \bepsilon_{s\beta}$ where each $\bepsilon_{s'} \in \{0,1\}^t$ (and $\beta = cA / t$).}
    \end{enumerate}
    \item Output $\by = \new{\left(\bepsilon_{\pi(1)} \circ \dots \circ \bepsilon_{\pi(B\beta)}\right)} \oplus \br.$
\end{enumerate}
\end{weirdFrame}

\begin{weirdFrame}{$\Dec_{\sk}^{\Tilde{\by}}(L,R)$}
{\setlength{\parindent}{0cm}\textbf{Input}: Interval $[L,R] \subseteq [k]$}

{\setlength{\parindent}{0cm}\textbf{Output}: word $\tilde{\bx}[L,R] \in \{0,1\}^{R - L + 1}.$}
\begin{enumerate}
    \item Parse $(\br,\pi) \leftarrow \sk$ and interpret \new{$\tilde{\by} = \tilde{\bepsilon}_{\pi(1)} \circ \dots \circ \tilde{\bepsilon}_{\pi(B\beta)}$ where each $\tilde{\bepsilon}_j \in \{0,1\}^t$}.
    \item Let the bits of $\bx[L,R]$ lie in blocks $\be_{s+1},\dots,\be_{s + \ell}.$ 
For each $j = s+1,\dots,s + \ell:$
\begin{enumerate}
    \item Compute $s'_1,\dots, s'_\beta$ such that \new{$\pi(s'_{r}) = (s + j - 1)\beta + r$ for all $r \in [\beta].$}
    \item {Query $\tilde{\by}$ to compute} $\tilde{\be}'_{j} \leftarrow \new{\tilde{\bepsilon}_{\pi(s'_1)} \circ \dots \circ \tilde{\bepsilon}_{\pi(s'_\beta)}}$.

    \item Compute $\tilde{\be}_j \leftarrow \Dec_{\calC}(\tilde{\be}'_{j})$.
\end{enumerate}
\end{enumerate}
From $\Tilde{\be}_{s + 1},\dots,\Tilde{\be}_{s + \ell}$ output bits corresponding to interval $[L,R]$. 
\end{weirdFrame}
\end{construction}
Any queries made by the decoder $\Dec$ are to the $t$-sized (corrupted) sub-blocks $\tilde{\bepsilon}_{\pi(r)}$, so Construction \ref{constr:paLDCv1BlockDecoding} is $t$-consecutive interval querying.
The main challenge is choosing the overall block size $cA$ such that decoding error remains negligible. 
We show that by increasing the original setting of block size $\tau$ in the work of Blocki and Zhang \cite{bloZha25} by a multiplicative $t$-factor, negligible decoding error follows.
\begin{theorem}\label{thm:b-block-dec-OTpaldc}
    Let $\lambda \in \N$. 
    In Construction \ref{constr:paLDCv1BlockDecoding}, suppose the $(A,a)$-code $\calC$ has (constant) rate $R$, (constant) error tolerance $\p \delta$, and (constant) alphabet $\Sigma$ with $c = \log |\Sigma|$. 
    Then, for any $t \in \N$ such that $t \mid cA$ and $k \in \poly(\lambda)$, $\p \calC = (\p \Gen, \p \Enc, \p \Dec)$ is a $t$-consecutive interval querying $(2 / R, O(a),\theta(\p \delta),\negl(\lambda),1)$-$\paLDC$ when $a = \omega(t\log \lambda)$.
\end{theorem}

\begin{proof}
We show that the probability of an incorrect decoding 
$\Pr[\paldcgame(\calA,\lambda,\alpha,\kappa,\delta,1) = 1]$ is negligible for $\delta = \theta(\p \delta)$ and any probabilistic polynomial time adversary $\calA$.
Define the event $\Bad = \bigcup_{j \in [k / ca]} \Bad_j$, where $\Bad_j$ is the event that $\be'_j$ has more than a $\p \delta$ fraction of errors. 
We show that the probability $\Pr[\Bad]$ is negligible.

Since the $t$-sized sub-blocks of the codeword are permuted and a random mask is applied, any errors added by a probabilistic polynomial time adversary $\calA$ to the corrupted codeword $\tilde{\by}$ are added uniformly at random over $t$-bit intervals.
This follows from generalizing the argument given by Lipton \cite{STACS:Lipton94}, where we interpret each sub-block as a $t$-bit symbol and the observation that the probability that event $\Bad$ occurs given that $\calA$ {\em does not} apply errors in a $t$-bit interval is at most the probability that event $\Bad$ occurs given that $\calA$ {\em does} apply errors in a $t$-bit interval.

Then, the number of errors in any given block $\be'_j$ follow a Hypergeometric$(m / t, \delta m / t, cA / t)$, which by the CDF bound of \cite{HemOstStraWoo11,hush_concentration_2005}, we have 
    \[\Pr[\mathtt{BAD}_{j}] < \exp\left(\frac{-2(((\p \delta - \delta)(cA /t))^2 - 1)}{(cA/t) + 1}\right).\]
Thus, for $\p \delta > \delta$ and $a/t \in \omega(\log n),$ this probability is negligible.
By a union bound, $\Pr[\Bad]$ is also negligible, so $\eps \leq \Pr[\Bad] < \negl(\lambda).$

The proof of the query parameters $\alpha = 2 / R$ and $\kappa = a$, is the same as in Theorem $16$ in the work of Zhang and Blocki \cite{bloZha25}. 
We repeat it here for completeness. 

For any $L,R \in [k]$ such that $R - L + 1 \geq a,$ suppose $\bx[L,R]$ lies in blocks $\be_{s + 1},\dots,\be_{s + \ell}$ for $\ell \leq \lfloor \frac{R - L  +1}{a}\rfloor + 1.$
To recover each of these $\be_{j}$ blocks, the decoder accesses the corresponding encoded blocks $\be_j'$ from $\tilde{y}.$
Thus, the total query complexity is $ cA\ell$  and so,
\begin{align*}
    \alpha &\leq \frac{cA\ell}{R - L  + 1} \\
           &\leq \frac{cA\left(\lfloor \frac{R - L + 1}{ca}\rfloor + 1\right)}{R - L + 1} \\
           &\leq \frac{1}{R} + \frac{cA}{ca} = \frac{2}{R} \tag{$ca \leq R - L + 1$}.
\end{align*}
\end{proof}
We note that the poly-round \paLDC \ construction of Blocki and Zhang \cite{bbgkz20} can also be made $t$-consecutive interval querying by the same technique of applying a higher-order permutation. 
We omit the construction and proof from this work since it will follow an almost-identical modification.

\subsection{The Ideal Insdel \paLDC \ Construction}\label{subsect:privateIdealInsdel}
We compile the $t$-consecutive interval querying, ideal Hamming \paLDC \ in the prior subsection into an ideal Insdel \paLDC \ using the Hamming-to-Insdel compiler \EncCompile \ in Section \ref{sect:bbgkz-modified}.
\begin{construction}\label{constr:finalINSDELpaldc}
Suppose $(m,k)$-code $\p \calC = (\p \Gen, \p \Enc, \p \Dec)$ is a $t$-consecutive interval querying Hamming $(\p\alpha,\p\kappa,\p\delta,\p\eps,q)$-\paLDC. 
Then define $\calC = (\Gen,\Enc,\Dec)$ as 
\begin{itemize}
    \item[] $\Gen(1^\lambda) := \p \Gen(1^\lambda)$,
    \item[] $\Enc_{\sk}(\bx) := \EncCompile({\Enc}_{\psk}(\bx)).$ 
    \item[] $\Dec_{\sk}^{\Tilde{\bY}}(L,R) := \Dec^{\Tilde{\bY}}_{\psk}(L,R)$
\end{itemize}
\end{construction}
We will need to show that the Hamming-to-Insdel compiler $\EncCompile$, when used within the private-key setting, retains correctness. 
We will argue a reduction from the correctness of the compiled Insdel \paLDC \ to the correctness of the underlying Hamming \paLDC. 
We first restate Theorem \ref{thm:compiler} in a convenient way for our reduction argument.
\begin{restatable}{corollary}{compilation}\label{corr:compilation}
    Let $(m,k)$-code $(\h \Enc, \h \Dec)$ be a $t$-consecutive interval querying Hamming $(\h \alpha, \h \kappa,\h \delta, \h \epsilon)$-\aLDC, and let $(n,k)$-code $(\Enc,\Dec)$ and $\EncCompile$ be defined as in Construction \ref{constr:modifiedBBGKZ}. 
    For any $\bx \in \{0,1\}^k$, $\tilde{\bY} \in \{0,1\}^{\tilde{n}}$ such that for $\bY = \Enc(\bx) \in \{0,1\}^n$, $\bY$ and $\tilde{\bY}$ are $\delta$-edit-close, and block size $\tau \in \Omega(\log n)$ such that $\tau \mid n$, let
    \[\simu {\tilde{\by}} \leftarrow \RecoverBlocks^{\tilde{\bY}}(1,\tilde{n} + 1,1,\lceil t / \tau \rceil) \circ \dots \circ \RecoverBlocks^{\tilde{\bY}}(1,\tilde{n} + 1,n / \tau - \lceil t / \tau \rceil + 1,n/\tau).\]
    Then, 
    \begin{enumerate}
        \item $\Pr[\Dec^{\tilde{\bY}}(L,R) \neq \bx[L,R]] \leq \Pr[\h \Dec^{\simu {\tilde{\by}}}(L,R) \neq \bx[L,R]] - \negl(n)$, and
        \item $\h \Enc(\bx)$ and $\simu {\tilde{\by}}$ are $\h \delta$-Hamming close, except with negligible probability.
        \item $\Dec^{\tilde{\bY}}$ on input interval $[L,R] \subseteq [k]$ of size $R - L + 1 \geq \h \kappa,$ makes $O\left( \frac{\h \alpha \tau \log^3 n}{t} \times (R - L)\right)$ queries to $\tilde{\bY}$. 
    \end{enumerate}
\end{restatable}

\begin{theorem}\label{thm:finalINSDELpaldc}
    Let $\p \calC = (\p \Gen, \p \Enc, \p \Dec)$ be a $(m,k)$-code that is a $t$-consecutive interval querying Hamming $(\p \alpha, \p \kappa, \p \delta, \p \eps)$-\paLDC. Then, $\calC = (\Gen,\Enc,\Dec)$ in Construction \ref{constr:finalINSDELpaldc} with is a $(n,k)$-code that is an Insdel $(\alpha,\kappa,\delta, \eps)$-\paLDC with $\alpha = O \left(\frac{\p \alpha \tau \log^3 n}{t}\right), \kappa = \p \kappa, \delta = \Omega(\p \delta),$ and $\eps = \p \eps + \negl(n).$
\end{theorem}

\begin{proof}
    $\alpha$ and $ \kappa$ parameters settings follow directly from Corollary \ref{corr:compilation}.
    
    Suppose for the sake of contradiction that there exists a probabilistic polynomial time adversary $\calA$ and a non-negligible function $\nu$ such that ${\Pr[\paldcgame_{\calC}(\calA,\lambda,\alpha,\kappa,\delta,q) = 1] > \nu(\lambda)},$ where $\paldcgame_{\calC}$ is the \paLDC \ decoding game over Insdel code $\calC$.
    We show that there exists an adversary $\calB$ for Hamming code $\p \calC$ i.e. there is a non-negligible function $\nu'$ such that 
    $\Pr[\paldcgame_{\p \calC}(\calB,\lambda,\p \alpha,\p \kappa,\p \delta,q) = 1] > \nu'(\lambda)$ where $\paldcgame_{\p \calC}$ is is the \paLDC \ decoding game over Hamming code $\p \calC$.
    The reduction is as follows:
    \begin{weirdFrame}{Reduction Argument}
    The challenger for $\paldcgame_{\p \calC}$ and adversary $\calB$ interact as follows:
        \begin{enumerate}
        \item The challenger generates secret key $\sk \leftarrow \p \Gen(1^\lambda)$. 
        \item $\calA$ outputs message $\bx \in \{0,1\}^k.$ $\calB$ sends $\bx$ to the challenger.
        \item The challenger sends back $\by \leftarrow \Enc_{\msf{p},\sk}(\bx)$ to adversary $\calB$. 
        \item Adversary $\calB$ computes $\bY \leftarrow \EncCompile(\by)$ and sends $\bY \in \{0,1\}^n$ to $\calA$.
        \item $\calA$ outputs $\tilde{\bY} \in \{0,1\}^{\tilde{n}}$ such that $\bY$ and $\tilde{\bY}$ are $\delta$-edit-close.
        By the assumption, there exists $[L,R] \subseteq [k]$ and non-negligible function $\nu'$ such that $R - L + 1 \geq \kappa$ and 
        \[\Pr\left[\Dec_{\sk}^{\tilde{\bY}}(L,R) \neq \bx[L,R]\right] > \nu'(\lambda).\]
        \item $\calB$ computes
       \[\simu {\tilde{\by}} \leftarrow \resizebox{.85\textwidth}{!}{$\RecoverBlocks^{\tilde{\bY}}(1,\tilde{n} + 1,1,\lceil t / \tau \rceil) \circ \dots \circ \RecoverBlocks^{\tilde{\bY}}(1,\tilde{n} + 1,n / \tau - \lceil t / \tau \rceil + 1,n/\tau)$}.\]
        \item $\calB$ outputs $\simu {\tilde{\by}}$.
        \end{enumerate}
    \end{weirdFrame}
    Let $\zeta_{\calA}$ be the event that $\paldcgame_{\calC}(\calA,\lambda,\alpha,\kappa,\delta,q) = 1$. 
    Likewise, $\zeta_{\calB}$ be the event that $\paldcgame_{\p \calC}(\calB,\lambda,\p \alpha,\p \kappa,\p \delta,q) = 1$.
    Then,  
    \begin{align*}
        \Pr\left[\zeta_{\calB} = 1 \right] &\geq \Pr[\zeta_{\calA} = 1] \Pr\left[\zeta_{\calB} = 1 \ \middle | \ \zeta_{\calA} = 1\right] \\ 
        &\geq \nu(\lambda) \times \Pr\left[\zeta_{ \calB} = 1\middle | \zeta_{\calA} = 1\right]
     \end{align*}
     It suffices to show that $\Pr\left[\zeta_{\calB} = 1\middle | \zeta_{\calA} = 1\right]$ is non-negligible.
      Set $\chi$ as the probability 
      \[\chi = \Pr\left[\Pr[\Dec_{\msf{p},\sk}^{\simuu {\tilde{\by}}}(L,R) \neq \bx[L,R] \ | \ \zeta_{\calA} = 1 ] \leq \negl(\lambda)\right].\]
      Then, by Corollary $\ref{corr:compilation}$ point 2, since $\tilde{\bY}$ are $\delta$-edit-close,
    \begin{align*}
         \Pr\left[\zeta_{\calB} = 1\middle | \zeta_{\calA} = 1\right] &\geq 1 - \left(\Pr\left[\Ham(\tilde{\by},\by) > \p \delta n \ \middle| \ \zeta_{\calA} = 1\right] + (1 - \chi)\right) \\
         &\geq \negl(\lambda) + (1 - \chi)
     \end{align*}
     Lastly, applying Corollary $\ref{corr:compilation}$ point 1 we can lowerbound $(1 - \chi)$,
     \begin{align*}
         (1 - \chi) \times 1&\geq \Pr[\Dec^{\simu{\tilde{\by}}}_{\msf{p},\sk} (L,R) \neq \bx[L,R]  ] \\
         &\geq \Pr[\Dec^{\tilde{\bY}}(L,R) \neq \bx[L,R]] + \negl(\lambda) \\
         &\geq \nu'(\lambda) + \negl(\lambda).
     \end{align*}
     Thus, $\Pr\left[\zeta_{\calB} = 1 \right] \geq \nu(\lambda) \times (1 - \chi)$ is non-negligible and so we reach a contradiction.
\end{proof}
By instantiating our private-key Insdel compiler with our $t$-querying Hamming \paLDC, we construct an ideal Insdel \paLDC.
\begin{corollary}\label{corr:finalINSDELpaldc}
    If $\p \calC$ of Construction \ref{constr:paLDCv1BlockDecoding}
    instantiated with a constant rate, constant error tolerant, and constant size alphabet code (e.g. a Justesen code), $t = \theta(\tau\log^3 n)$ and $\tau = \theta(\log n)$,
    then code $\calC$ in Construction \ref{constr:finalINSDELpaldc} is an (Ideal) Insdel $(O(1),O(\log^5 n), O(1), \negl(n))$-\paLDC \ with constant rate.
\end{corollary} 
\section{Ideal Insdel \aLDC s for Resource-bounded Channels }\label{sect:resource-bounded}
In this section, we present an ideal insdel \aLDC \ in settings where the channel is bounded for some resource, such as parallel time or circuit depth.
As in the construction of an ideal Insdel \paLDC, our construction for resource-bounded channels will rely on a construction of a consecutive interval querying Hamming \aLDC \ that will allow the compiler to amortize over its queries. 
Then, we formally prove the Compiler we constructed in Section \ref{sect:bbgkz-modified} is secure in resource-bounded settings, which results in our ideal Insdel construction for resource-bounded channels.

We start by recalling the definition of an \aLDC\ for resource-bounded channels (\raldc).
\begin{definition}[\raldc \ \cite{SCN:AmeBloBlo22}]
    A $(n,k)$-code $\calC = (\Enc,\Dec)$ is a $(\alpha, \kappa, \delta, \epsilon,\R)$-resource-bounded amortizeable LDC (\raldc)  for hamming errors (resp. insDel errors) if for any algorithm $\calA$ in algorithm class $\R$, there is a negligible function $\mu$ such that
        \[\Pr[\raldcgame(\calA,\lambda,\delta,\kappa) = 1] \leq \mu(\lambda),\]
        where the probability is taken over the randomness of $\calA, \Gen,$ and $\paldcgame$. 
        The experiment \raldcgame\ is defined as follows:
     \begin{weirdFrame}
        {\raldcgame$(\calA,\lambda,\delta,\kappa)$}
        \begin{enumerate}
            \item The adversary $\calA$ chooses a message $\bx \in \{0,1\}^k$ and sends it to the encoder.
            \item The encoder sends $\by \leftarrow \Enc_{\sk}(\bx)$ to the adversary.
            \item The adversary outputs $\Tilde{\by} \in \{0,1\}^*$ that is $\delta$-hamming-close (resp. $\delta$-edit-close) distance to $\by^{(h)}$.
            \item If there exists $L^{(h)}, R^{(h)} \in [k]$ such that $R^{(h)} - L^{(h)} + 1 \geq \kappa$ and
            \[\hspace{-\leftmargin}\Pr\left[\Dec^{\Tilde{\by}}(L,R) \neq \bx[L,R]\right] > \eps(\lambda)\] such that $\Dec^{\Tilde{\by}}(.)$ makes at most $(R - L + 1) \alpha$ queries to $\tilde \by$, then this experiment outputs $1$.
        \end{enumerate}
        Otherwise, output $0.$
    \end{weirdFrame}
\end{definition}
\subsection{A consecutive interval querying, Ideal \raldc}\label{subsect:resourceIdealHamming}
To construct a $t$-consecutive interval querying, ideal \raldc , we employ the Hamming \paLDC \ to \raldc \ compiler of Blocki and Zhang \cite{bloZha25}, which was a modification of the original compiler by Ameri et al. \cite{SCN:AmeBloBlo22}.
The construction makes use of two other building blocks: the first is cryptographic puzzles, consisting of two algorithms $\PuzzGen$ and $\PuzzSolve.$
On input seed $\bc$, $\PuzzGen$ outputs a puzzle $\bZ$ with solution is $\bc$ i.e. $\PuzzSolve(\bZ) = \bc.$ 
The security requirement states that adversary $\calA$ in a defined algorithm class $\R,$ cannot solve the puzzle $\bZ$ and cannot even distinguish between tuples $(\bZ_0,\bc_0,\bc_1)$ and $(\bZ_1,\bc_0,\bc_1)$ for random $\bc_i$ and $\bZ_i = \PuzzGen(\bc_i).$
Ameri et al. are able to construct {\em memory-hard} puzzles, i.e. cryptographic puzzles unsolvable by algorithm class $\R_{\msf{cmc}}$ of algorithms with bounded cumulative memory complexity, under standard cryptographic assumptions. 
We generalize their definition for any class of algorithms $\R$, where if a cryptographic puzzle is unsolvable by any algorithm in $\R$, we say the puzzle is $\R$-hard.
\begin{definition}[\Puzz \ \cite{SCN:AmeBloBlo22}]\label{def:puz}
    A puzzle $\Puzz = (\PuzzGen,\Sol)$ is a $\R$-hard puzzle for algorithm class $\R$ if there exists a polynomial $t'$ such that for all polynomials $t > t'$ and every algorithm $\calA \in \R,$ there exists $\lambda_0$ such that for al $\lambda > \lambda_0$ and every $s_0,s_1 \in \{0,1\}^\lambda$, we have 
    \[\left| \Pr \calA(Z_{b},Z_{1 - b},s_0,s_1) = b] - \frac{1}{2}\right| \leq \negl(\lambda),\]
    where the probability is taken over $b \unif \{0,1\}$ and $Z_i \leftarrow \Gen(1^\lambda,t(\lambda),s_i)$ for $i \in \{0,1\}.$
\end{definition}
The second building block is a variant of a locally decodable code, referred to as $\LDC^*$, which recovers the entire (short) encoded message $\bs$ with locality only scaling linearly with the message length. 
\begin{definition}[\LdcStar \cite{ITC:BloKulZho20}]\label{def:ldcStar}
    A $(n,k)$-code $\calC = (\Enc,\Dec)$ is an $(\ell,\delta,\eps)$-\LdcStar\ if for all $\bc \in \Sigma^k$, $\Dec^{\tilde{\by_*}}$ with query access to word $\tilde\by_*$ $\delta$-close to $\Enc(\bc)$, 
    \[\Pr[\Dec^{\tilde{\by_*}} = \bc] \geq (1 - \eps),\]
    where $\Dec^{\tilde{\by_*}}$ makes at most $\ell$ queries to $\tilde{\by_*}$.
\end{definition}
The original $\LDC^*$ construction given by Blocki et al. \cite{ITC:BloKulZho20} is a repetition code of a constant rate, constant error tolerance code (e.g. the Justesen code), where $\Dec^*$ queries a constant number of repeated codewords and performs a majority vote decoding.
On message length $k$, their \LdcStar decoder makes $\theta(k)$ queries, where the codeword length $n \gg k$ can be arbitrarily large.

Given a $\paLDC$ code $\p \calC = (\p \Gen, \p \Enc, \p \Dec)$ and a $\LDC^*$ code $\calC_* = (\Enc_*,\Dec_*)$, the constructed Hamming $\raldc$ encoder $\Enc$, on message $\bx$ will (1) generate the secret key $\sk \leftarrow \p \Gen(1^\lambda; \bc)$ with some random coins $\bc$ (2) Compute the \paLDC \ encoding of the message $\p \by \leftarrow \Enc_{\psk}(\bx)$ (3) Compute the $\LDC^*$ encoding of the puzzle $\by_* \leftarrow \Enc_*(\bZ)$, where the $\LDC^*$ encoding has codeword length $\theta(|\p \by |)$ (4) and finally, output $\by = \p \by \circ \by_*$.
Intuitively, a resource-bounded channel will not be able to solve the puzzle to retrieve the secret key from $\by_*$.
Hence, the message encoding $\p \by$ looks random to the channel, reducing the resource-bounded channel's view of the codeword to the view a computationally-bound channel in the private-key setting.
On the other hand, the $\raldc$ decoder $\Dec$ will have sufficient resources to solve the puzzle and locally decode the codeword on input interval $[L,R]$ i.e. $\Dec$ computes $\tilde{\bc} \leftarrow \PuzzSolve(\Dec_*(\tilde{\by_*}))$, $\tilde{\sk} \leftarrow \p \Gen(1^\lambda; \tilde{\bc}),$ and outputs $\Dec_{\msf{p} \tilde{\sk}}^{\tilde{\p \by}}(L,R)$.
By choosing an ideal \paLDC \ code $\p \calC$ and an appropriate $\LDC^*$ code $\calC_*$, the resulting compiled code $(\Enc,\Dec)$ is an ideal $\raldc$.

We observe that the compiled \raldc \ decoder \Dec\ satisfies $t$-consecutive interval query when instantiated with a $t$-consecutive interval querying \paLDC\ and the original $\LDC^*$ construction by Blocki et al. \cite{ITC:BloKulZho20}.  
First, observe that the decoder $\Dec$ queries the message encoding $\p \by$ and the secret-key randomness encoding $\by_*$ disjointly using $\p \Dec$ and $\Dec_*$ respectively. 
Second, $\Dec_*$ is also $t$-consecutive interval querying as long as the encoding used in the repetition encoding has codeword length $t.$ 
Thus, the overall \raldc\ decoder \Dec is $t$-consecutive interval querying.
We summarize this result formally below.
\begin{theorem}\label{thm:paLdctoResourceTblock}
Let $\lambda \in \N$. Let $\p \calC = (\p \Gen, \p \Enc, \p \Dec)$ be a $(\p n,\p k)$-code that is a $t$-consecutive interval querying Hamming $(\p \alpha, \p \kappa, \p \delta, \p \eps)$-\paLDC. For any algorithm class $\R$ such that there exists $\R$-hard cryptographic puzzles, there exists a $(\res n, \res k)$-code $\res \calC = (\res \Enc,\res \Dec)$ that is a $t$-consecutive interval querying Hamming $(\res \alpha,\res \kappa,\res \delta,\res \eps,\R)$-\raldc\ for $\res n = \theta(\p n),\res k = \p k = \poly(\lambda),\res \alpha = \theta(\p \alpha), \res \kappa = \p \kappa, \res \delta = \theta( \p \delta),$ and $\res \eps = \p \eps + \negl(n)$.
\end{theorem}

\subsection{The Ideal Insdel \raldc\ Construction}\label{subsect:resourceIdealInsdel}
We present our final construction of an ideal Insdel \raldc . 
We will proceed similarly to our construction of an ideal Insdel \paLDC \ in Section \ref{sect:comp-bounded}, where we use our Hamming-to-Insdel compiler \EncCompile\ with the ideal, $t$-consecutive interval querying Hamming \raldc\ constructed in the prior subsection.
\begin{construction}\label{constr:finalINSDELraldc}
Suppose $(m,k)$-code $\res \calC = (\res \Enc, \res \Dec)$ is a $t$-consecutive interval querying Hamming {$(\res \alpha,\res \kappa,\res \delta,\res \eps,\R)$}-\raldc. 
Then define $\calC = (\Enc,\Dec)$ as 
\begin{itemize}
    \item[] $\Enc(\bx) := \EncCompile(\res \Enc(\bx)).$ 
    \item[] $\Dec^{\Tilde{\bY}}(L,R) := \res \Dec^{\tilde{\bY}}(L,R)$
\end{itemize}
\end{construction}
Just as before, we will need to show that security holds when we use the Hamming-to-Insdel compiler in a resource-bounded setting by giving a reduction argument. 
For convenience, we restate Corollary \ref{corr:compilation}.
\compilation*
One additional nuance we will need to handle is to allow any adversary for the Hamming code sufficient resources to perform the reduction.
Recall that in order to reach a contradiction, the adversary $\calB$ performed two non-resource-trivial computations: (1) adversary $\calB$ ran $\EncCompile(\by)$ and (2) simulated the Hamming codeword by computing
\[\simu {\tilde{\by}} \leftarrow \RecoverBlocks^{\tilde{\bY}}(1,\tilde{n} + 1,1,t) \circ \dots \circ \RecoverBlocks^{\tilde{\bY}}(1,\tilde{n} + 1,B - t + 1,B).\]
Let $T(n)$ (resp. $S(n)$) be the running time (resp. space usage) of the computation of  $\EncCompile(\by)$ and $\simu{\tilde{\by}}$ .
Let $\calB = \reduce_n(\calA)$ be a reduction from algorithm $\calA$ to $\calB$ in in time $T(n)$ time with space usage $S(n)$.
Then, for any class of algorithms $\R(n),$ denote its closure class $\overline{\R}(n)$ with respect to $\reduce_n$ defined as the minimum class of algorithms such that $\reduce_n(\calA) \in \overline{\R}(n)$ for all $\calA \in \R(n).$
Note that we can define a more precise closure by considering the reduction with respect to the {\em parallel} running time since each $\RecoverBlocks$ call to construct $\simu \by$ can be ran in parallel.
Additionally note that for most classes $\R,$ the closure class is the same as the original class i.e. $\overline{\R} = \R$, with an example being the class of cumulative memory complexity bound algorithms.

\begin{theorem}\label{thm:finalINSDELraldc}
    Let $\res \calC = (\res \Enc, \res \Dec)$ be a $(m,k)$-code that is a $t$-consecutive interval querying Hamming $(\res \alpha,\res \kappa,\res \delta,\res \eps,\R(n))$-\raldc. Then, $\calC = (\Enc,\Dec)$ in Construction \ref{constr:finalINSDELraldc} is a $(n,k)$-code that is an Insdel $(\alpha,\kappa,\delta, \eps,\overline{\R}(n))$-\raldc \ with $\alpha = O \left(\frac{\res \alpha \tau \log^3 n}{t}\right), \kappa = \res \kappa, \delta = \Omega(\res \delta),$ and $\eps = \res \eps + \negl(n).$
\end{theorem}
\begin{proof}
    Suppose for the sake of contradiction that there exists an adversary $\calA \in \R(n)$ and a non-negligible function $\nu$ such that \[{\Pr[\raldcgame_{\calC}(\calA,\lambda,\alpha,\kappa,\delta,q) = 1] > \nu(\lambda)},\] where $\raldcgame_{\calC}$ is the \raldc \ decoding game over Insdel code $\calC$.
    We show that there exists an adversary $\calB \in \overline{R}$ for Hamming code $\res \calC$ i.e. there is a non-negligible function $\nu'$ such that 
    $\Pr[\raldcgame_{\res \calC}(\calB,\lambda,\res \alpha,\res \kappa,\res \delta,q) = 1] > \nu'(\lambda)$ where $\raldcgame_{\res \calC}$ is is the \raldc \ decoding game over Hamming code $\res \calC$.
    Adversary $\calB$ is defined the same as in the proof of Theorem \ref{thm:finalINSDELpaldc}, and since by construction, the adversary $\calB \in \overline{\R}(n),$ we reach a contradiction.
\end{proof}
\begin{corollary}\label{corr:finalINSDELraldc}
    Let $t = \theta(\tau\log^3 n)$, $\tau = \theta(\log n)$, and define any algorithm class $\R$ such that there exists $\R$-hard cryptographic puzzles.
    Then by Theorem \ref{thm:paLdctoResourceTblock}, code $\calC$ in Construction \ref{constr:finalINSDELraldc} is an Insdel $(O(1),O(\log^5 n), O(1), \negl(n),\R)$-\raldc\ with constant rate.
\end{corollary}

\bibliographystyle{IEEEtran}
\bibliography{bib/other,bib/abbrev0,bib/crypto}

\appendix 

\section{Expanded Discussion of Related Work} 
\label{app:relatedWork}
\paragraph*{\LDC \ Lowerbounds for Worst-case Errors}
Locally decodable codes for worst-case errors have been extensively studied. 
In particular, ideal (constant rate, constant error-tolerance, and constant locality) \LDC s for worst-case errors do not exist.
Any \LDC \ with constant locality $\ell \geq 2$ and constant error tolerance $\delta > 0$ must have sub-constant rate \cite{STOC:KatTre00}, where the best known constructions have super-polynomial codeword length (e.g., matching vector codes) \cite{LDCSURVEY}.
In the special case of locality $\ell = 2$, any construction necessarily has exponential rate \cite{STOC:KerdWo03,CCC:GKST06,FOCS:BenRegdWo08}. 
This lower bound is matched by the Hadamard code.

Further, Katz and Trevisan show that there do not exist \LDC\ for locality $\ell < 2$ even when the rate is allowed to be exponential\cite{STOC:KatTre00} or when the error-rate is relaxed to be sub-constant, i.e., $\delta = o(1)$.
In other words, a local decoder must read {\em at least twice} as much information as requested in the setting of worst-case errors.




\paragraph*{Relaxed Locally Decodable Codes}
To deal with the undesirable rate, locality, and error-tolerance trade-offs for worst-case error \LDC s, Ben-Sasson et al. introduced the notion of a relaxed \LDC s, allowing the decoder to reject (output $\bot$) instead of outputting a codeword symbol whenever it detects error \cite{STOC:BGHSV04}.
This was further expanded by Gur et al. to study locally {\em correctable} codes (\LCC s) where the decoder returns symbols of the codeword instead of the message \cite{ITCS:GurRamRot18}.
However, ideal relaxed \LDC/\LCC s are still provably unachievable. 
Gur and Lachish prove that any prove that any relaxed \LDC/\LCC \ with constant rate and constant error-tolerance must have $\Tilde{\Omega}(\sqrt{\log k})$ locality \cite{SODA:GurLac20}\footnote{The result of Gur and Lachish also implies that any relaxed \LDC\ with constant locality and error-tolerance must have codeword length $n = \Omega\left(k^{1 + c}\right)$, where $c = 1 / O\left(\ell^2\right)$}, ruling out any possibility of having constant rate, constant locality, and constant error-tolerance.
Prior to this lowerbound, Gur et al. construct relaxed \LCC s with constant rate and constant error-tolerance but with sub-optimal, super-polylogarithmic locality $\ell = (\log k)^{O(\log\log k)}$ \cite{SODA:ChiGurShi20, STOC:BGHSV04}, raising the question of whether the lowerbound of Gur and Lachish was achievable. 
Recently, Kumar and Mon made progress towards reaching this lowerbound by constructing the first relaxed \LCC  s with constant rate, constant error-tolerance, and $\polylog (k)$ locality \cite{STOC:KumMon24}. 
Kumar and Mon built off a recent breakthrough in the construction of asymptotically optimal locally testable codes \cite{STOC:DELLM22}. 
Their code construction essentially concatenated smaller locally testable codewords in a block-like fashion such that the overall codeword retained (relaxed) local correctability. 
At a high level, the local decoder will apply the local tester recursively to detect if the block corresponding to the requested index has sufficiently high error: if so, the local decoder outputs $\bot$, and if not, recovers the entire block.  
Kumar and Mon's construction achieves constant rate, constant error tolerance, and $O((\log k)^{69})$ locality.
Their construction was further improved to locality $\ell = (\log k)^{2 + o(1)}$ by Cohen and Yankovitz \cite{CCC:CohYan24} by swapping the locally testable code used in the construction with a expander code.
Nonetheless, the existence of a constant rate, constant error-tolerance, and $\Tilde{\Omega}(\sqrt{\log k})$ locality \LCC \ remains an open question.
Additionally, it is unknown whether one can construct an ideal {\em amortized} relaxed \LDC.

\paragraph*{\LDC s for Computationally/Resource-bounded Channels}
Another popular \LDC \ relaxation allow a sender and receiver to share randomness that is unknown to the error channel. 
This was first introduced by Ostrovsky et al. under the notion of private locally decodable codes \cite{EPRINT:OstPanSah07}. 
Ostrovsky et al. first give a construction with constant rate, constant error-tolerant, and polylog locality for arbitrary error channels when the shared randomness (the secret key) is only used once. 
When the secret key is used a polynomial number of times, Ostrovsky et al. provide a construction with an additional log factor in the locality under the additional assumption that the error channel can only make polynomial time computations. 

Private locally decodable codes were extended to {\em public-key} locally decodable codes, where the sender, receiver, and channel share public randomness, by Hemenway and Ostrovsky \cite{C:HemOst08}.
They show that a public-key locally decodable code can be constructed with a Private Information Retrieval (PIR) scheme and a IND-CPA public-key encryption. 
Hemenway et al. improved on this construction by lowering the key-length to be linear to the security parameter and improving the locality by a square-root factor \cite{HemOstStraWoo11}.

While these codes depend on shared randomness, this assumption can be removed if the error channel is assumed to be resource-bounded, e.g., the channel cannot evaluate circuits beyond a particular size/depth. 
Blocki et al. provide the first framework for converting any private \LDC \ into an \LDC \ for resource-bound channels under the assumption of a function computable by the sender and receiver but not the channel \cite{ITC:BloKulZho20}.
Alternatively, Ameri et al. provide a similar converting framework under the assumption of hard puzzles \cite{SCN:AmeBloBlo22}.



\paragraph*{Batch Codes and PIR}
{Information-theoretic Private Information Retrieval (PIR) schemes with multiple servers can be constructed from (smooth) locally decodable codes. 
Ishai et al. \cite{STOC:IKOS04} introduce the notion of a batch code which can be used to minimize the amortized workload of any individual PIR server. 
More formally, a $(k,N,\kappa,m)$ batch code encodes a message $\bx \in \Sigma^k$ into a $m$-tuple $\by_1,\dots,\by_m \in \Sigma^*$ of total length $N$ such that for any set of indices $i_1,\dots,i_\kappa \in [k],$ the symbols $x_{i_1},\ldots, x_{i_\kappa}$ can be decoded by reading {\em at most one symbol} from each bucket.}

{Ishai et al. show how to convert a smooth \LDC \ into a batch code. 
Informally, a \LDC\ is smooth if its decoder queries any symbol of the (possibly corrupted) codeword with equal probability. 
A smooth \LDC \ can be generically transformed into a batch code, where the encoding scheme is unchanged and the decoding scheme iteratively makes oracle calls to the \LDC\ smooth decoder.}

{However, it is unclear whether a batch code can be used to construct amortized Hamming \LDC s. 
Since the definitions differ substantially e.g., batch codes are not required to be error-tolerant, we do not expect a transformation from a batch code to a \aLDC\ without significant rate or locality blowup.
Conversely, it is also unclear whether amortized \LDC s can be used to construct batch codes or whether amortized \LDC s can be used to obtain more efficient PIR schemes.
}

\end{document}